\pdfoutput=1

\documentclass[journal, final, cspaper, compsoc]{IEEEtran}

\usepackage{multirow}
\usepackage{amsfonts,amsmath,amssymb,amsthm}
\makeatletter
\renewcommand*\env@matrix[1][*\c@MaxMatrixCols c]{%
  \hskip -\arraycolsep
  \let\@ifnextchar\new@ifnextchar
  \array{#1}}
\makeatother
\usepackage{graphicx}
\usepackage{subcaption}
\usepackage{url}
\usepackage{color}

\newcommand{\be}{\begin{eqnarray}}
\newcommand{\ee}{\end{eqnarray}}

\newcommand{\ben}{\begin{enumerate}}
\newcommand{\een}{\end{enumerate}}
\newcommand{\beq}{\begin{equation}}
\newcommand{\eeq}{\end{equation}}
\newcommand{\beqa}{\begin{eqnarray*}}
\newcommand{\eeqa}{\end{eqnarray*}}
\newcommand{\bit}{\begin{itemize}}
\newcommand{\eit}{\end{itemize}}
\newcommand{\bt}{\begin{tabular}{c}}
\newcommand{\btt}{\begin{tabular}}
\newcommand{\et}{\end{tabular}}

\newtheorem{definition}{Definition}

\newtheorem{corollary}{Corollary}
\newtheorem{lemma}{Lemma}
\newtheorem{theorem}{Theorem}

\newtheorem{assumption}{Assumption}
\newtheorem{problem}{Problem}

\DeclareMathOperator*{\argmin}{arg\,min}
\DeclareMathOperator*{\argmax}{arg\,max}

\newcommand{\squishlist}{
   \begin{list}{$\bullet$}
    { \setlength{\itemsep}{0pt}      \setlength{\parsep}{0pt}
      \setlength{\topsep}{3pt}       \setlength{\partopsep}{0pt}
      \setlength{\listparindent}{-2pt}
      \setlength{\itemindent}{-5pt}
      \setlength{\leftmargin}{1em} \setlength{\labelwidth}{0em}
      \setlength{\labelsep}{0.5em} } }
\newcommand{\squishend}{
    \end{list}  }

\def\ps@IEEEtitlepagestyle{%
  \def\@oddfoot{\mycopyrightnotice}%
  \def\@evenfoot{}%
}
\def\mycopyrightnotice{%
  {\footnotesize The copyright is held by author/owner(s).}%
  \gdef\mycopyrightnotice{}
}

\IEEEpubid{\mycopyrightnotice˜}

\begin{document}

\title{An Efficient and Fair Multi-Resource Allocation Mechanism for Heterogeneous Servers}

\author{
\IEEEauthorblockN{Jalal Khamse-Ashari\IEEEauthorrefmark{1}, Ioannis Lambadaris\IEEEauthorrefmark{1}, George Kesidis\IEEEauthorrefmark{2}, Bhuvan Urgaonkar\IEEEauthorrefmark{2} and Yiqiang Zhao\IEEEauthorrefmark{3}\\}
\IEEEauthorblockA{\IEEEauthorrefmark{1}Dept. of Systems and Computer Engineering, Carleton University, Ottawa, Canada\\
\IEEEauthorrefmark{2}School of EECS, Pennsylvania State University, State College, PA, USA\\
\IEEEauthorrefmark{3}School of Math and Statistics, Carleton University, Ottawa, Canada\\
Emails: \IEEEauthorrefmark{1}\{jalalkhamseashari,ioannis\}@sce.carleton.ca,~\IEEEauthorrefmark{2}\{gik2,buu1\}@psu.edu
~\IEEEauthorrefmark{3}zhao@math.carleton.ca}}

%

%
%
\maketitle

\begin{abstract}
Efficient and fair allocation of multiple types of resources is a crucial objective in a cloud/distributed computing cluster.
Users may have diverse resource needs. Furthermore, diversity in server properties/capabilities may mean that only a subset of servers may be usable by a given user. In platforms with such heterogeneity, we identify important limitations in existing multi-resource fair allocation mechanisms, notably Dominant Resource Fairness (DRF) and its follow-up work.
To overcome such limitations,
we propose a new \emph{server-based approach}; each server allocates resources by maximizing a per-server \emph{utility function}.
We propose a specific class of utility functions which, when appropriately parameterized,
adjusts the trade-off between efficiency and fairness, and captures a variety of fairness measures
(such as our recently proposed Per-Server Dominant Share Fairness).
We establish conditions for the proposed mechanism to satisfy certain properties that are generally deemed desirable, e.g.,
envy-freeness, sharing incentive, bottleneck fairness, and Pareto optimality.
To implement our resource allocation mechanism, we develop an iterative algorithm which is shown to be globally convergent.
Finally, we show how the proposed mechanism could be implemented in a distributed fashion.
We carry out extensive trace-driven simulations to show the enhanced performance of our proposed mechanism over the existing ones.
\end{abstract}


\section{Introduction}

Cloud computing has become increasingly popular as it provides a cost-effective alternative to proprietary high performance computing
systems. As the workloads to data-centers housing cloud computing platforms are intensively growing,
developing an efficient and fair allocation mechanism which guarantees quality-of-service for different workloads has become increasingly important.
Efficient and fair resource allocation in such a shared computing system is particularly challenging because of
(a) the presence of multiple types of resources, (b) diversity in the workloads' needs for these resources, (c) heterogeneity in the resource capacities of servers, and (d) placement constraints on which servers may be used by a workload.
In the following four paragraphs we briefly elaborate on each of these complexities.


The \emph{multi-resource needs} of cloud workloads imply that conventional single-resource oriented notions of fairness are inadequate \cite{DRF}.
Dominant Resource Fairness (DRF) is the first allocation mechanism which describes a notion of fairness for allocating multiple types of resources for a single server system. 
Using DRF users receive a \emph{fair share} of their \emph{dominant resource} \cite{DRF}.
Of all the resources requested by the user (for every unit of work called a \emph{task}), its dominant resource is the one with the highest demand when demands are expressed as fractions of the overall resource capacities. DRF is shown to achieve several properties that are
commonly considered desirable from a multi-resource \emph{fair} allocation mechanism.

\emph{Heterogeneity of workloads' resource demands} is another complexity which results in a trade-off between efficiency and fairness.
Specifically, heterogeneity of users' demands may preclude some resources from being fully utilized. Hence, the DRF allocation may result in a poor resource utilization even when there is only one server~\cite{Chiang12,Grandl,bonald2014}. To address this issue, \cite{Chiang12} proposed to allocate resources by applying the so-called $\alpha$-proportional fairness (instead of max-min fairness~\cite{DN}) on dominant shares. The proposed mechanism, when appropriately parameterized, 
adjusts the trade-off between efficiency and fairness.
However, it is applicable only to a \emph{single} server/resource-pool.


In the case of \emph{multiple heterogeneous servers}, there are several studies investigating/extending DRF allocation when there is \emph{no placement constraint}
\cite{CDRF,DRFH15,HUG}.
In all of these works, fairness is defined in terms of a \emph{global metric}, a scalar parameter defined in terms of different resources across all servers. E.g., \cite{DRFH15} presents an extension to DRF where the dominant resource for each user is identified as if all resources were concatenated at one server, and subsequently the resources are allocated by applying max-min fairness on the dominant shares.
Since such a global metric may not perfectly capture the impact of server heterogeneity, such approaches may lead to an inefficient resource utilization (see Section~\ref{sec:main:challenge} for further discussions and Section~\ref{sec:bg:psdsf} for an illustrative example).
Moreover, such mechanisms may not be readily implementable in a distributed fashion~\cite{zhu2015}, as each server needs information on the available resources over all servers. Such information may not be available at each server, especially in a cloud computing environment where the resource capacities (and even activity of servers) might be churning.

There are limited works in the literature investigating multi-resource fair allocation in the presence of \emph{user placement constraints} \cite{UDRF, TSF}.
In this case, it is yet unclear how to globally identify the dominant resource as well as the dominant share for different users, as each one may have access only to a subset of servers. Work in \cite{TSF} presents an extension to DRF identifying the user share by ignoring placement constraints and applying a similar approach as in an unconstrained setting. We show that this approach may not achieve fairness in the specific case that one of the resources serves as a bottleneck (see Section \ref{sec:main:challenge}).

In \cite{PSDSF} we proposed a multi-resource \emph{fair} allocation mechanism, called PS-DSF, which is applicable to heterogeneous servers in the presence of placement constraints.
The intuition behind PS-DSF is to capture the impact of server heterogeneity by measuring the total allocated resources to each user explicitly from the perspective of each server. Specifically, PS-DSF identifies a virtual dominant share (VDS) for each user {\em with respect to each server} (as opposed to a single system-wide dominant share in DRF).
The VDS for user $n$ with respect to server $i$ is defined as the ratio of $x_n$ - the total number of tasks allocated to user $n$ -
over the number of tasks executable by user $n$ when monopolizing server $i$.
Then the resources at each server are allocated by applying max-min fairness on VDS (see Section~\ref{sec:bg:psdsf} for a detailed discussion).
This approach is amenable to a distributed implementation. It results in an enhanced performance over the existing mechanisms,
and satisfies certain properties essential for fair allocation of resources \cite{PSDSF}.


%

\subsection{Contributions}
In this paper, we build upon and generalize our proposed PS-DSF allocation mechanism \cite{PSDSF}
to \emph{capture the trade-off between efficiency and fairness}. We concisely summarize our contributions.
\squishlist
\item
    We propose a new \emph{server-based formulation} (which includes PS-DSF as a special case) to allocate resources while capturing server heterogeneity. The new formulation can be viewed as a \emph{concave game} among different servers,
    where each server allocates resources by maximizing a \emph{per-server} utility function (Section~\ref{sec:main:frm}).
\item
    We study a specific class of utility functions which results in an extension of $\alpha$-proportional fairness on VDS.
    We show how the resulting allocation, which we call $\alpha$PF-VDS, captures the trade-off between efficiency and fairness by adjusting the  parameter $\alpha$. We show that $\alpha$PF-VDS satisfies bottleneck fairness, envy-freeness and sharing incentive properties (as defined in Section~\ref{sec:bg:drf}) for $\alpha\ge1$, and Pareto optimality for $\alpha=1$ (Section~\ref{sec:main:Def} and \ref{sec:main:properties}).
\item
    We develop a (centralized) convergent algorithm to implement our proposed mechanism.
    Towards this, we introduce an equivalent formulation 
    for which we derive an iterative solution (Section~\ref{sec:reformulate} and \ref{sec:algo:ctr}).
\item
    We propose a simple heuristic to develop a distributed implementation for our resource allocation mechanism (Section~\ref{sec:algo:dtr}).
\item We carry-out extensive simulations, driven by real-world traces, to show the enhanced performance of our proposed mechanism (Section~\ref{sec:eval}).
\squishend

\subsection{Related Work}
{\bf Resource allocation with a game-theoretic approach.}
There are several works in the literature which study the resource allocation problem in a cloud computing environment
with a game-theoretic approach~\cite{bredin2000, jalaparti2010, wei2010, xu2014, zhu2016, zhu2015EoF}.
Among these, \cite{bredin2000, jalaparti2010, wei2010} are limited to a \emph{single-resource} setting,
while \cite{xu2014, zhu2016, zhu2015EoF} consider a \emph{multi-resource} environment.
In these studies the multi-resource allocation problem is formulated as a \emph{game}, where players are different servers.
Since they choose DRF as the underlying notion of fairness, they will have the same limitations as DRF for heterogeneous servers
(see Section~\ref{sec:main:challenge} for a discussion of such limitations).
Moreover, some of them need to solve an extensive form game with a huge strategy set space, e.g., \cite{xu2014}, which may not be implementable in a distributed fashion.

{\bf Scheduling in the presence of placement constraints.}
There are some recent works investigating max-min fair allocation/scheduling for one type of resource while respecting placement constraints \cite{Ghodsi13,midrr,CM4FQ, Ashari17, Ashari2017j}. These single-resource schedulers could be useful in a multi-resource setting when one of the resources is dominantly requested by all users. Otherwise, they might result in a poor resource utilization \cite{DRF,Ghodsi13}.

\section{Model and background work}\label{sec:bg}
Consider a set $\mathcal{K}$ of ${K}$ heterogeneous servers/resource-pools\footnote{A resource-pool may consist of a group of \emph{homogeneous servers}.} each containing $M$ types of resources. We denote by $c_{i,r}\ge0$, the capacity (i.e., amount) of resource $r$ $(1,2,\cdots,M)$ on server $i$.
We make the reasonable assumption that all resources on each server are arbitrarily divisible among the users running on it.
Let $\mathcal{N}$ denote the set of $N$ active users.
Let $\phi_n>0$ denote the weight associated with user $n$. The weights reflect the priority of users with respect to each other.
Let ${\bf d}_n=[d_{n,r}]$ denote the per task \emph{demand vector} for user $n\in{\mathcal N}$,
i.e., the amount of each resource required for executing one task for user $n$.
Let $x_{n,i}\in\mathbb{R}^+$ denote the number of tasks that are allocated to user $n$ from server $i$.
Assuming linearly proportionate resource-needs\footnote{The assumption of linearly proportionate resource needs is admittedly an idealization. We are following convention set by DRF and used by follow up works.}, $x_{n,i}{\bf d}_n=[x_{n,i}d_{n,r}]$ gives the amounts of different resources demanded by user $n$ from server $i$.


Due to heterogeneity of users and servers, each user may be restricted to get service \emph{only from a subset of servers}.
For example, users may not run tasks on servers which lack some required resources.
Furthermore, each user may have some special hardware/software requirements
(e.g., public IP address, a particular kernel version, GPU, etc.)
which further restrict the set of servers that the user's tasks may run on.
Let $\mathcal{N}_i\neq\emptyset$ denote the set of \emph{eligible users} for server $i$.
The placement constraints imply that $x_{n,i}=0,~n\notin\mathcal{N}_i,~\forall i$.

For instance, consider the example in Fig.~\ref{fig:example1}, where three types of resources, CPU, RAM, and network bandwidth are available over two servers in the amounts of ${\bf c}_1=$[12 cores,~4GB,~75Mb/s] and ${\bf c}_2=$[8 cores,~16GB,~0Mb/s], where no communication bandwidth is available over the second server; four users with their corresponding demand vectors are also shown in the figure. In this example, the first two users require network bandwidth for execution of their tasks, so they are not \emph{eligible} to run tasks on the second server. However, the last two users 
may run tasks on both servers.

\begin{figure}
    \centering
\includegraphics[width=3 in]{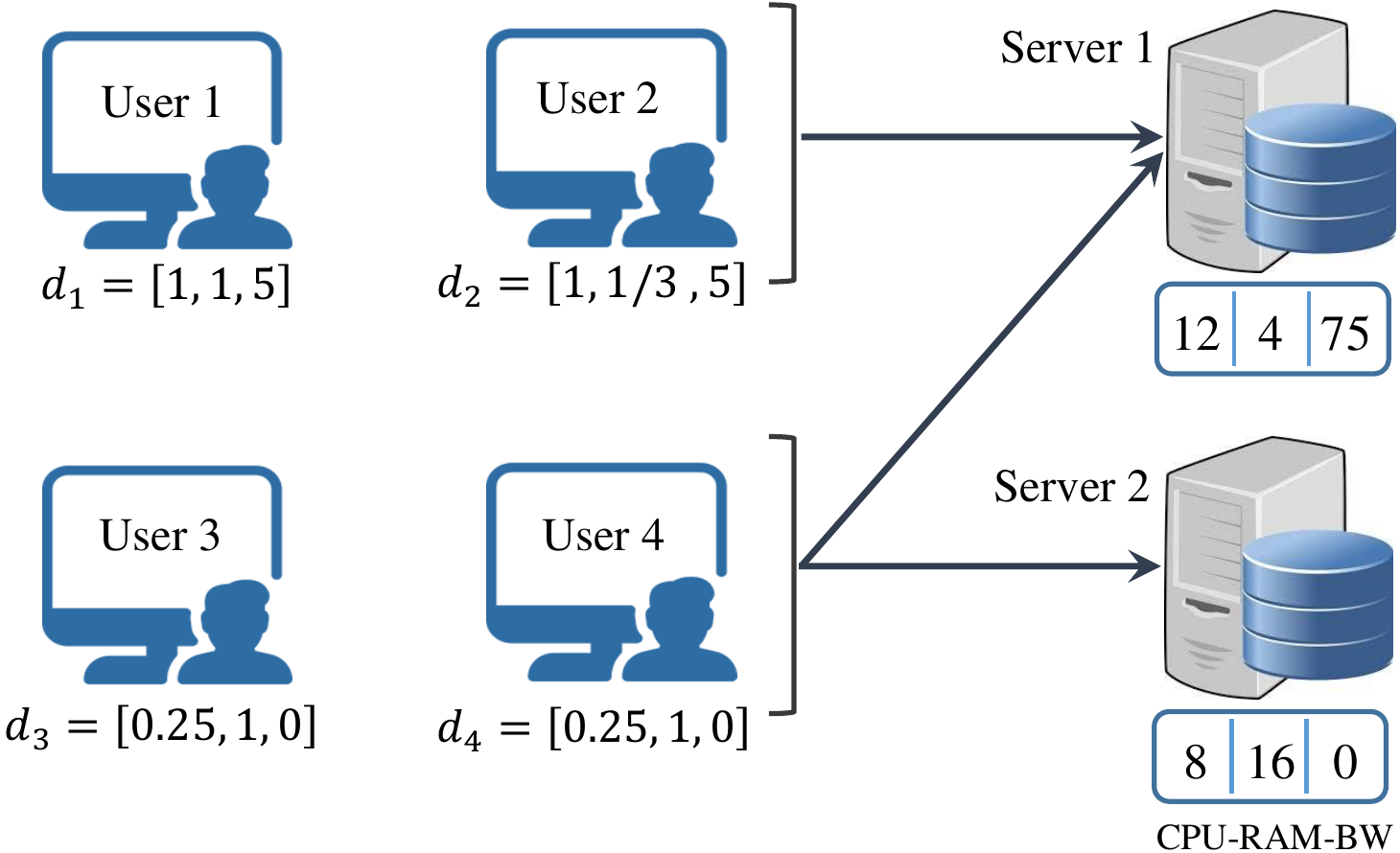}
\centering
\caption{\footnotesize A heterogeneous multi-resource system with two servers and four equally weighted users.}
\label{fig:example1}
\end{figure}

\subsection{Dominant resource fairness}\label{sec:bg:drf}
Multi-resource fair allocation was originally studied in \cite{DRF} under the assumption that all resources are aggregated at one resource-pool. Specifically, let $c_r$ denote the total capacity of resource $r$.
Let ${\bf a}_n=[a_{n,r}]$ denote the amounts of different resources allocated to user $n$ under some allocation mechanism. The utilization of user $n$ of its allocated resources, $U_n({\bf a}_n)$, is defined as the number of tasks,
$x_n$, which could be executed using ${\bf a}_n$, that is:
\be\label{utility}
U_n({\bf a}_n) \triangleq x_n = \min_r\frac{a_{n,r}}{d_{n,r}}.
\ee
In \cite{DRF} the following properties are deemed desirable for a multi-resource allocation mechanism.
\squishlist
\item \emph{Sharing incentive:}
    Each user is able to run more tasks compared to a \emph{uniform allocation} where each user $n$ is allocated a $\phi_n/\sum_m\phi_m$ fraction of each resource.
\item \emph{Envy freeness:} A user should not prefer the allocation vector of another user when adjusted according to their weights, i.e.,
    it should hold that $U_n({\bf a}_n)\ge U_n(\frac{\phi_n}{\phi_m}{\bf a}_m)$ for all $n,m$.
\item \emph{Bottleneck fairness:} If there is one resource which is \emph{dominantly requested by every user},
    then the allocation satisfies \emph{max-min fairness} for that resource.
\item \emph{Pareto optimality:} It should not be possible to increase the number of tasks $x_n$ for any user $n$, without decreasing $x_m$ for some other user(s).
\item \emph{Strategy proofness:} Users should not be able to increase their utilization by erroneously declaring their resource demands.
\squishend


The reader is referred to \cite{DRF} or \cite{DRF12} for further details.
Sharing incentive provides some sort of performance isolation,
as it guarantees a minimum utilization for each user
irrespective of the demands of the other users.
Envy freeness embodies the notion of fairness.
Bottleneck fairness describes a necessary condition which applies to a specific case that
one resource is dominantly requested by every user, so that a \emph{single-resource} notion of fairness is applicable.
These three properties are essential to achieve fairness.
So, we refer to them as \emph{essential fairness-related} properties.
Pareto optimality is a benchmark for maximizing system utilization.
Finally, strategy proofness prevents users from gaming the allocation mechanism.
In our view these properties are applicable mainly for private settings. In public settings, users pay explicit costs for their usage or allocations and the provider's goal is to maximize its profits subject to allocation guarantees for users. Even for private clouds, strategy proofness would only be necessary in settings where users act selfishly. In many private settings, users are cooperative and here strategy proofness is not needed. In view of this, we will not consider strategy proofness.

DRF is the first multi-resource allocation mechanism satisfying all the above properties.
Specifically, for every user $n$, the \emph{Dominant Resource} (DR) is defined as \cite{DRF}:
\be\label{DR}
\rho(n):=\argmax_rd_{n,r}/c_r,
\ee
that is, the resource whose greatest portion is required for execution of one task for user $n$.
The fraction of the DR that is allocated to user $n$ is defined as its \emph{dominant share}:
\be\label{DS}
s_n:=\frac{a_{n,\rho(n)}}{c_{\rho(n)}}.
\ee

Without loss of generality, we may restrict ourselves to non-wasteful allocations, i.e., ${\bf a}_n=x_n{\bf d}_n,~\forall n$.
Hence, an allocation $\{x_n\}$ is feasible when:
\begin{eqnarray}
&&\sum_nx_nd_{n,r}\le c_r,~\forall r.
\end{eqnarray}
\begin{definition}\label{DRF_Def}
An allocation $\{x_n\}$ satisfies \emph{DRF}, if it is feasible and the weighted dominant share for each user, $s_n/\phi_n$ cannot be increased while
maintaining feasibility without decreasing $s_m$ for some user $m$ with $s_m/\phi_m\le s_n/\phi_n$ \cite{DRF}.
\end{definition}

DRF is a restatement of \emph{max-min fairness} in terms of \emph{dominant shares}.
What make it appealing are the desirable properties which are satisfied under this allocation mechanism.

\subsection{Existing challenges with heterogeneous servers and placement constraints}\label{sec:main:challenge}
In case of heterogeneous servers (whether there are any placement constraints or not),
a natural approach to extend DRF is to identify a system-wide dominant resource for each user, \emph{as if} all resources were
concatenated within a \emph{single virtual server}. Specifically, let $c_r:=\sum_ic_{i,r}$ denote the total capacity of resource $r$ within such a virtual server. Then, one may identify the dominant resource for each user $n$ according to \eqref{DR}. Furthermore, the \emph{global dominant share} for user $n$ is given by:
\be
s_n=x_n\max_r\frac{d_{n,r}}{c_{r}},\label{dominant_share}
\ee
where $x_n$ is the total number of tasks that are allocated to user $n$ from different servers, that is $x_n:=\sum_ix_{n,i}$.
As in Definition~\ref{DRF_Def}, one may find an allocation $\{x_{n,i}\}$ which satisfies max-min fairness in terms of the global dominant shares~\cite{DRFH15}. Such as allocation mechanism, referred to as DRFH, is shown to achieve \emph{Pareto optimality} and \emph{envy freeness}. However, it \emph{fails} to provide \emph{sharing incentive}~\cite{DRFH15}.
We believe that the definition of bottleneck fairness employed by DRFH (with respect to a single virtual server that aggregates all resources) is also controversial. Specifically, if all users have the same dominant resource (with respect to the above mentioned virtual server), then DRFH satisfies max-min fairness with respect to such a resource~\cite{DRFH15}. 
In case of heterogeneous servers with placement constraints, however, one may consider other conditions under which a resource serves as a \emph{bottleneck}.


\begin{definition}\label{BF_def}
A resource $\rho$ is said to be a \emph{bottleneck} if for every server $i$:
\be
\frac{d_{n,\rho}}{c_{i,\rho}}\ge\frac{d_{n,r}}{c_{i,r}},~\forall r,~n\in\mathcal{N}_i.\label{BF_criteria}
\ee
If there exists a bottleneck resource, then the allocation should satisfy max-min fairness with respect to that resource.
\end{definition}

Unfortunately, DRFH does not satisfy bottleneck fairness in the sense of Definition~\ref{BF_def}.
To appreciate this shortcoming of the DRFH mechanism, consider the example in Fig.~\ref{fig:example1},
where the second resource (RAM) is dominantly requested by eligible users at each server.
According to Definition~\ref{BF_def}, RAM is identified as the bottleneck resource in this example.
To allocate the RAM resources in a fair manner, each user should be allocated $x_1=x_{1,1}=2$, $x_2=x_{2,1}=6$, $x_3=x_{3,2}=8$ and $x_4=x_{4,2}=8$ tasks, respectively (This allocation results from our proposed PS-DSF allocation mechanism~\cite{PSDSF}). On the other hand, the DRFH mechanism would instead identify network bandwidth as the dominant resource for the first two users and RAM as the dominant resource for the last two users. To achieve max-min fairness in terms of dominant shares, the DRFH mechanism allocates  $x_1=x_2=3$ and $x_3=x_4=8$ tasks to each user. Under such an allocation, the RAM resources are not allocated in a fair manner to the first two users.

Yet another extension of DRF, which applies to heterogeneous servers in the presence of placement constraints, is TSF~\cite{TSF}.
As in~\cite{TSF}, we let $\gamma_{n,i}$ denote the number of tasks that user $n$ may execute when monopolizing server $i$ (i.e., when $n$ is the only user). Let $\gamma_n:=\sum_i\gamma_{n,i}$ be defined as the number of tasks executable for user $n$
when monopolizing all servers as \emph{if} there were no placement constraints.
An allocation is said to satisfy Task Share Fairness (TSF), when $x_n/\gamma_n$ satisfies max-min fairness~\cite{TSF}.
When there is only one server, then $x_n/\gamma_n$ results in the dominant share for each user $n$.
In such case, TSF reduces to DRF. In case of heterogeneous servers with placement constraints,
TSF is shown to satisfy Pareto optimality, envy freeness and sharing incentive properties~\cite{TSF}.
However, we show by example that this mechanism may not 
satisfy bottleneck fairness (neither in the sense of Definition~\ref{BF_def},
nor in the conventional sense based on considering a single virtual server introduced above~\cite{PSDSF}).

For instance, consider again the example in Fig.~\ref{fig:example1},
where the second resource is identified as a bottleneck according to Definition~\ref{BF_def}.
The number of tasks that each user may run in the whole cluster is $\gamma_1=4$, $\gamma_2=12$,  and $\gamma_3=\gamma_4=4+16=20$ tasks, respectively.
Hence, each user is allocated $x_1=x_{1,1}=5/3$, $x_2=x_{2,1}=5$, $x_3=x_{3,1}+x_{3,2}=8+1/3=25/3$ and $x_4=x_{4,1}+x_{4,2}=8+1/3=25/3$ tasks, according to the TSF mechanism, which differs from the fair allocation in this example.


\subsection{Per-server dominant share fairness (PS-DSF)}\label{sec:bg:psdsf}
In this subsection, we describe PS-DSF which we introduced in \cite{PSDSF}.
PS-DSF is an extension to DRF which is applicable for heterogeneous servers in the presence of placement constraints.
The core idea of this mechanism is to introduce a ``\emph{virtual dominant share}" for every user,
with respect to each server.
Towards this, we first identify the dominant resource for every user $n$ with respect to each server $i$,
\be
\rho(n,i):=\argmax_r\frac{d_{n,r}}{c_{i,r}}.
\ee
Let $\gamma_{n,i}$ denote the number of tasks which could be executed by user $n\in\mathcal{N}_i$
when monopolizes server $i$,
\be
\gamma_{n,i}:=\min_r\frac{c_{i,r}}{d_{n,r}}=\frac{c_{i,\rho(n,i)}}{d_{n,\rho(n,i)}},~n\in\mathcal{N}_i.
\ee
It is assumed that $\gamma_{n,i}>0$ for all $n\in\mathcal{N}_i$. We set $\gamma_{n,i}=0$ if $n\notin\mathcal{N}_i$.
\begin{definition}\label{Def_VDS}
The Virtual Dominant Share (VDS) for user $n$ with respect to server $i$, $s_{n,i}$, is defined as:
\be
s_{n,i} := \frac{x_n}{\gamma_{n,i}} = \frac{x_nd_{n,\rho(n,i)}}{c_{i,\rho(n,i)}},\label{VDS}
\ee
where $x_n=\sum_jx_{n,j}$ is the \emph{total number of tasks} that are allocated to user $n$
(whether or not these tasks are actually allocated using server $i$).
\end{definition}
We have the following conditions on an allocation, ${\bf x}:=\{x_{n,i}\in\mathbb{R}^+\mid n\in\mathcal{N},~i\in\mathcal{K}\}$, to be feasible:
\be
&&\sum_{n\in\mathcal{N}_i}x_{n,i}d_{n,r}\le c_{i,r},~\forall i,r.\qquad\quad\label{FC1_1}\\
&& x_{n,i}=0,~n\notin\mathcal{N}_i,~\forall i. \label{FC1_2}
\ee
\begin{definition}\label{PS_DSF_Def}
An allocation ${\bf x}$ satisfies PS-DSF, if it is feasible and the allocated tasks to each user, $x_n$ cannot be increased (while maintaining feasibility) without decreasing $x_{m,i}$ for some user $m$ and server $i$ with $s_{m,i}/\phi_m\le s_{n,i}/\phi_n$.
\end{definition}

Intuitively, $s_{n,i}$ gives the normalized share of the dominant resource for user $n$ with respect to server $i$ which should be allocated to it as if $x_n$ tasks were allocated resources solely from server $i$ (see the right hand side of \eqref{VDS}). The reader may note that $s_{n,i}$ could be possibly greater than 1, as some tasks might be allocated to user $n$ from other servers. According to PS-DSF, the available resources at each server $i$ are allocated by applying (weighted) max-min fairness on $\{s_{n,i}\}$.
It can be seen that PS-DSF reduces to DRF when there is only one server.

\begin{figure}
    \centering
\includegraphics[width=0.99\columnwidth]{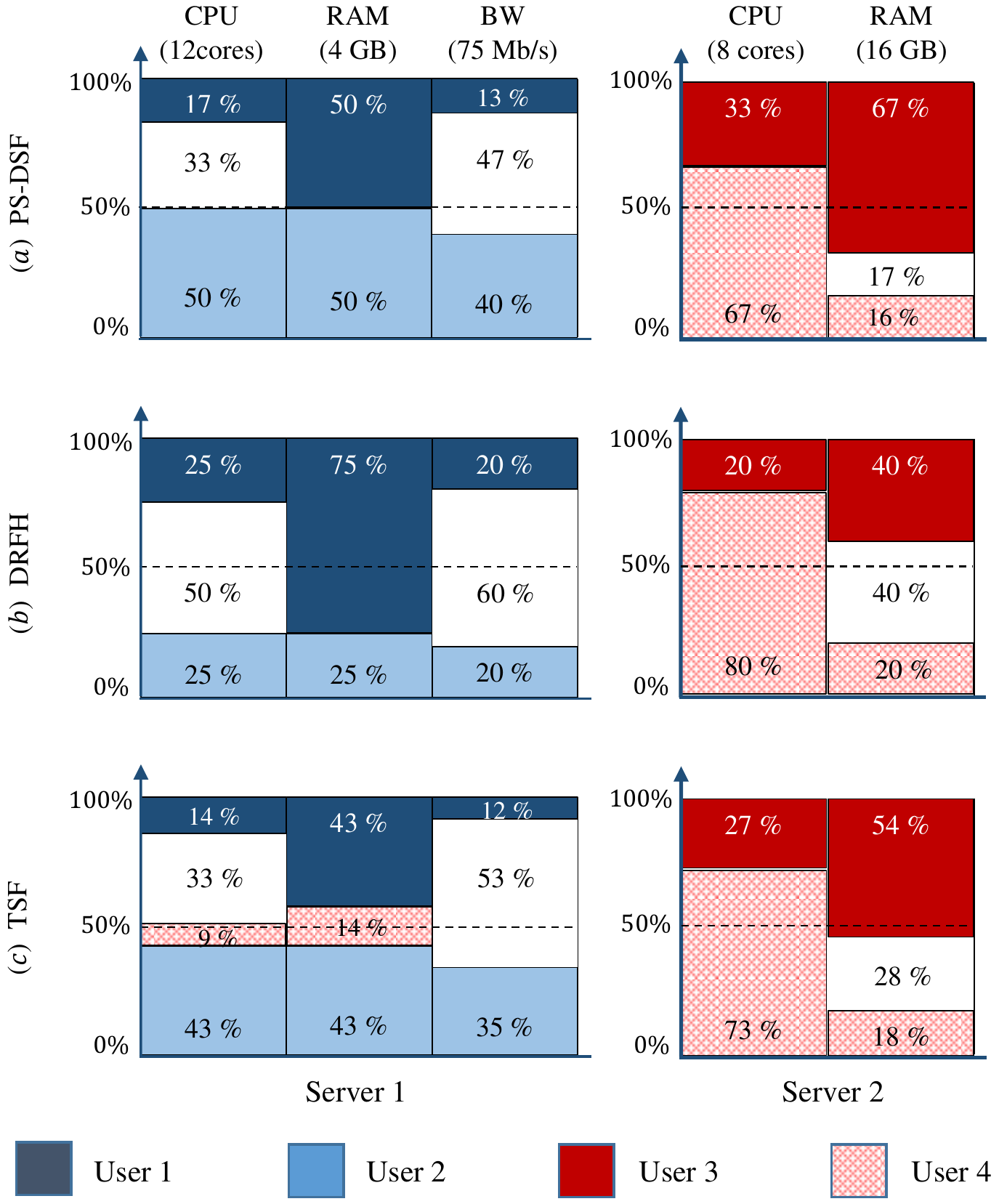}
\centering
\caption{\footnotesize Comparing the PS-DSF allocation with the DRFH and TSF allocations.
The PS-DSF allocation mechanism is more efficient in utilizing different resources.}
\label{fig:example2}
\end{figure}

To gain more intuition, consider again the example in Fig.~\ref{fig:example1}, but this time let $d_4=[1,0.5,0]$.
In this case, each user may run $\gamma_{1,1}=4$, $\gamma_{2,1}=12$, $\gamma_{3,1}=4$, $\gamma_{4,1}=8$ tasks when monopolizing server 1. The third and the fourth users each may run $\gamma_{3,2}=16$ and $\gamma_{4,2}=8$ tasks when monopolizing server 2.
In order to satisfy PS-DSF, each user should be allocated $x_1=x_{1,1}=2$, $x_2=x_{2,1}=6$, $x_3=x_{3,2}=32/3$ and $x_4=x_{4,2}=16/3$ tasks, respectively. Therefore, the VDS (c.f. Definition~\ref{Def_VDS}) for each user with respect to the first server is $s_{1,1}=s_{2,1}=0.5$, $s_{3,1}=8/3$ and $s_{4,1}=2/3$.
Also, the VDS for user 3 and 4 with respect to the second server is $s_{3,2}=s_{4,2}=2/3$.
The reader can verify that for each server $i$ the allocated tasks to any
user may not be increased without decreasing the allocated tasks to another user with a less or equal VDS.
The resulting PS-DSF allocation is shown in Fig.~\ref{fig:example2}.
The DRFH and TSF allocations for this example are also illustrated in Fig.~\ref{fig:example2}.
It can be seen that the PS-DSF allocation mechanism is more efficient in utilizing different resources compared to the DRFH and TSF mechanisms.

\begin{table}[ht]
\footnotesize
\centering
\captionsetup{name=Table }
\caption{\footnotesize Properties of different allocation mechanisms in case of heterogeneous servers with placement constraints: sharing~incentive (SI), envy~freeness (EF), Pareto optimality (PO), and bottleneck fairness (BF).}
\label{table1}
\begin{tabular}{|c||c|c|c|}
\hline
Property      &      ~~DRFH         &     ~~TSF~~        &       PS-DSF\\
\hline
\hline
SI            &                     &  $\checkmark$  &      $\checkmark$\\
\hline
EF            &   $\checkmark$      &  $\checkmark$  &      $\checkmark$\\
\hline
PO            &   $\checkmark$      &  $\checkmark$  &      \\
\hline
BF            &                     &                &      $\checkmark$\\
\hline
\end{tabular}
\end{table}


The reader may note that PS-DSF does not satisfy Pareto optimality in general.
It is worth noting that Pareto optimality may not also be satisfied in other works, e.g., \cite{zhu2016, zhu2015EoF},
which aim at developing a distributed implementation for DRFH. PS-DSF not only is amenable to distributed implementation (as we show in~\cite{PSDSF}),
but also may lead to more efficient utilization of resources compared to the DRFH and TSF mechanisms~\cite{PSDSF} (as also can be observed in Fig.~\ref{fig:example2}, or in the trace-driven simulations in Section~\ref{sec:eval}).
The intuitive reason for this is that each of the DRFH and TSF allocation mechanisms allocates resources based on a \emph{global metric}.
Since a global metric throws away information about the actual distribution of resources across servers, approaches based on it may not perfectly capture the impact of server heterogeneity, and therefore may lead to an inefficient resource utilization in heterogeneous settings.



In summary, PS-DSF has been shown to satisfy the essential fairness-related properties, i.e., envy-freeness, sharing incentive and bottleneck fairness, has been observed to offer highly efficient utilization of resources,
and is amenable to distributed implementation~\cite{PSDSF}. 


\section{A server-based approach for multi-resource allocation}\label{sec:main}
As already discussed, 
in most of the existing multi-resource allocation mechanisms,
fairness is defined in terms of a global metric,
a scalar parameter defined for each user in terms of different resources across all servers.
Such mechanisms may not succeed in satisfying all the essential fairness-related properties (c.f. Section~\ref{sec:main:challenge}), may not readily be implementable in a distributed fashion, and may lead to inefficient resource utilization.
In this section, we propose a new formulation for multi-resource allocation problem which is based on a \emph{per-server metric} (as opposed to a global metric) for different users, so that server heterogeneity is captured. The proposed allocation mechanism is built upon our proposed PS-DSF allocation mechanism~\cite{PSDSF}, which was briefly described in the previous section. It generalizes PS-DSF in order to address the trade-off between efficiency and fairness. Furthermore, it inherits all the properties that are satisfied by PS-DSF.

\subsection{Problem formulation}\label{sec:main:frm}
As defined in Section~\ref{sec:bg:psdsf}, the VDS is a per-server metric
which gives a measure of the allocated resources to each user from the perspective of each server.
According to PS-DSF, the available resources at each server are allocated by applying \emph{max-min fairness}
on VDS.  In order to address the trade-off between efficiency and fairness,
we may choose to allocate resources at each server by applying the so-called \emph{$\alpha$-proportional fairness} \cite{Walrand} on VDS.
To this end, we propose a general formulation, where each server $i$ strives to maximize a ``\emph{per-server utility function}".
Specifically, each server $i$ tries to find an allocation ${\bf x}_i :=[x_{n,i}]$ which solves the following problem\footnote{The utility of each server, as defined by \eqref{P1_1}, depends on its-own allocation/action, ${\bf x}_i$, as well as actions taken by other servers, ${\bf x}_{-i}:=\{{\bf x}_j\mid j\neq i\}$. This is the standard notation used in the context of game theory.}.
\begin{problem}\label{P1} For every server $i$:
\vspace{-1mm}
\be
&& \max_{{\bf x}_i} U_i({\bf x}_i,{\bf x}_{-i}):=\sum_{n\in\mathcal{N}_i}\phi_ng_i\left(\frac{x_n}{\phi_n\gamma_{n,i}}\right)\label{P1_1}\quad\\
&& \text{Subject to:}~\sum_{n\in\mathcal{N}_i}x_{n,i}d_{n,r}\le c_{i,r},~\forall r,\label{P1_2}\\
&& \text{~~~~~~~~~~~~~~}x_{n,i}\ge0,~\forall n\in\mathcal{N}_i,\label{P1_3}\\
&& \text{~~~~~~~~~~~~~~}x_{n,i} = 0,~\forall n\notin\mathcal{N}_i,\label{P1_4}
\ee
where $x_n=\sum_jx_{n,j}$, and $g_i(\cdot)$, as can be also seen in \cite{Walrand}, is a scalar function, which is twice-differentiable, strictly concave, and increasing.
\end{problem}

\begin{definition}
An allocation ${\bf x}$ is said to be \emph{feasible} if it satisfies the feasibility conditions in \eqref{P1_2}-\eqref{P1_4} for all servers.
\end{definition}

In Section \ref{sec:main:Def}, we will present specific choices for $g_i(\cdot)$,
which capture the \emph{trade-off between efficiency and fairness},
and span a variety of allocations,
including the so-called \emph{proportional fair allocation}, and the PS-DSF allocation.

Solving Problem~\ref{P1} concurrently over different servers is a game,
where each server strives to maximize its-own utility.
In fact, Problem~\ref{P1} describes a \emph{concave game} whose players are different servers.
It is well-known that a Nash Equilibrium\footnote{A feasible allocation $({\bf x}^*_i,{\bf x}^*_{-i})$ is a Nash Equilibrium if no unilateral deviation in action by any single server/player is profitable for that server. That is, $\forall i,~\forall \text{ (feasible) } {\bf x}_i:~ U_i({\bf x}^*_i,{\bf x}^*_{-i})\ge U_i({\bf x}_i,{\bf x}^*_{-i})$}
(NE) always exists for such a concave game~\cite{Rosen}.
Further discussions on the structure of the solution set (Nash equilibriums), and 
some conditions governing uniqueness of the solution,
will be described in Section~\ref{sec:reformulate}.

\subsection{$\alpha$-proportional fairness on virtual dominant shares}\label{sec:main:Def}
At the optimal solution(s) to Problem~\ref{P1},
 \emph{not all} capacity constraints may be active.
In fact, there exist trade-offs between efficiency and fairness,
which depend on the specific choice of $g_i(\cdot)$.
To capture the trade-off between efficiency and fairness,
one may choose $g_i(\cdot)$ from the class of $\alpha$-fair utility functions~\cite{Walrand}.
Specifically, we choose $g_i(z)$ such that $g'_i(z)=z^{-\alpha}$, for some fixed parameter $\alpha$.
For this class of utility functions, the optimal solution to Problem~\ref{P1} satisfies an extension of $\alpha$-proportional fairness in terms of virtual dominant shares, which we call ``\emph{$\alpha$-Proportional Fairness on VDS}", or in short $\alpha$PF-VDS.

\begin{definition}
A feasible allocation, ${\bf x}$, satisfies $\alpha$PF-VDS, if for every feasible allocation ${\bf y}$, and for every server $i$:
\be
\sum_{n\in\mathcal{N}_i}\frac{(y_{n,i}-x_{n,i})/\gamma_{n,i}}{\tilde{s}_{n,i}^\alpha}\le0\label{fair_criteria},
\ee
where $\tilde{s}_{n,i}:=s_{n,i}/\phi_n=x_n/\gamma_{n,i}\phi_n$ is the weighted VDS for user $n$ with respect to server $i$.
\end{definition}

\begin{theorem}\label{Th_basic}
Let $g_i(z)$ be from the class of $\alpha$-fair
utility functions with $g'_i(z) = z^{-\alpha}$, $\alpha > 0$.
A feasible allocation, ${\bf x}$, is a solution to Problem~\ref{P1} if and only if it satisfies $\alpha$PF-VDS.
\end{theorem}

The proof is given in the appendix.
The following theorem, again proven in the appendix, describes how $\alpha$PF-VDS is related to other notions of fairness.

\begin{theorem}\label{Th_psdsf}
The $\alpha$PF-VDS allocation is weighted proportionally fair\footnote{An allocation, ${\bf x}$,
is weighted proportionally fair if it is feasible, and if for any other feasible allocation, ${\bf y}$,
the weighted summation of proportional changes is not positive, i.e., $\sum_n\phi_n\frac{{y}_n-x_n}{x_n}\le0$ (see \cite{kelly1998rate}).} for $\alpha=1$, and approaches a PS-DSF allocation as $\alpha\rightarrow\infty$.
\end{theorem}


Consider again the example in Fig.~\ref{fig:example1}, but let $d_4=[1,0.5,0]$.
In this example, $\alpha$PF-VDS results in the same allocation at server $1$, for every $\alpha>0$.
In other words, the $\alpha$PF-VDS allocation coincides with the PS-DSF allocation at server 1, for every $\alpha>0$.
The reason is that RAM is dominantly requested by both of users 1 and 2 which are allocated resources using server~1 (see Corollary~\ref{Coro_universal}).
The $\alpha$PF-VDS allocation for server 2 is depicted in Fig.~\ref{fig:example3}, for $\alpha=1$ (proportional fair allocation), $\alpha=3$, and $\alpha=\infty$ (PS-DSF allocation). It can be observed that the proportional fair allocation is more efficient in utilizing different resources. However, eligible users for this server tend to get the same portion of their respective dominant resources as $\alpha\rightarrow\infty$.
This shows how $\alpha$-PF-VDS captures the tradeoff between efficiency and fairness.

\begin{figure}
    \centering
\includegraphics[width=0.98\columnwidth]{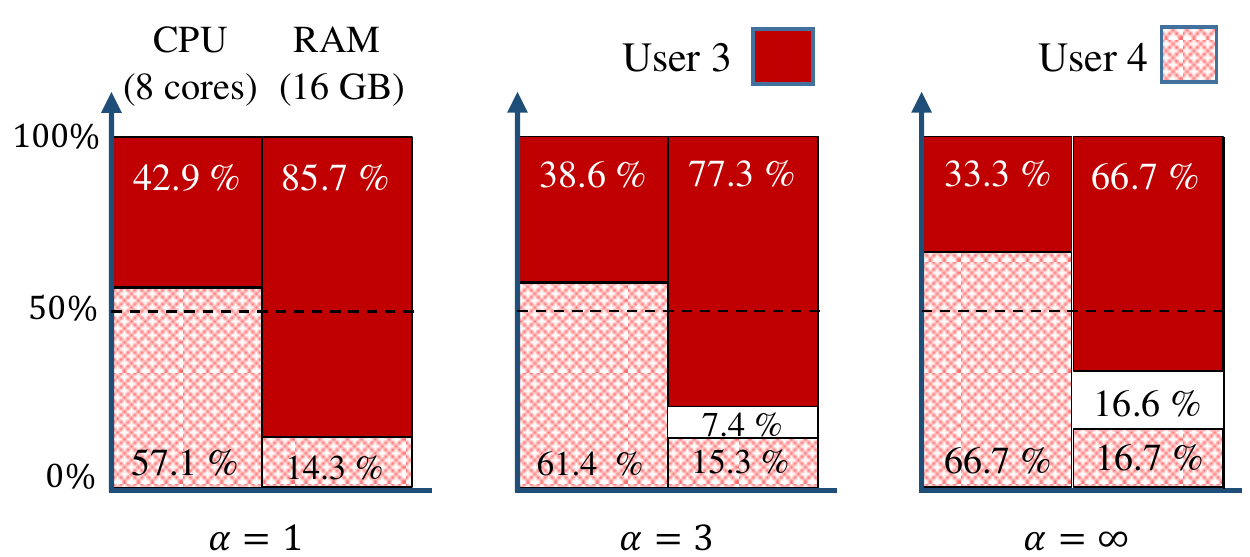}
\centering
\caption{\footnotesize
An illustration of how the $\alpha$PF-VDS allocation mechanism
can be parameterized (via $\alpha$) to capture the tradeoff between efficiency and fairness.}
\label{fig:example3}
\vspace{-2mm}
\end{figure}

\subsection{The properties of the $\alpha$PF-VDS allocation mechanism}\label{sec:main:properties}
In this section, we investigate different properties which are satisfied by the $\alpha$PF-VDS mechanism.
In case of heterogeneous servers with placement constraints,
we need to extend the notion of sharing incentive property.
The notion of bottleneck fairness has been extended by Definition~\ref{BF_def}.
Other properties, Pareto optimality and envy freeness follow the same definitions as described in Section~\ref{sec:bg}.
To generalize the sharing incentive property, consider a \emph{uniform allocation}, where a fraction $\phi_n/\sum_m\phi_m$ of the available resources over each server (whether this server is eligible or not) is allocated to each user $n$.
An allocation is said to satisfy \emph{sharing incentive}, when each user is able to run more tasks compared to such a uniform allocation.

\begin{theorem}\label{Th_Properties}
The $\alpha$PF-VDS allocation mechanism satisfies envy-freeness and sharing incentive properties for every $\alpha\ge1$.
It also satisfies Pareto optimality for $\alpha=1$.
\end{theorem}

In case that all resources are integrated at one server,
\cite{Chiang12} has proposed an allocation mechanism which applies
$\alpha$-proportional Fairness on Dominant Shares (FDS).
It can be observed that $\alpha$PF-VDS reduces to
$\alpha$-proportional FDS when there exists only one server.
In \cite{Chiang12} it is shown that sharing incentive and envy freeness properties are not necessarily satisfied under $\alpha$-proportional FDS when $\alpha<1$. Hence, $\alpha$PF-VDS may also violate sharing incentive and/or envy freeness properties for $\alpha<1$.

The last property that we consider in this section is bottleneck fairness.
According to Definition~\ref{BF_def}, a resource is considered as a \emph{bottleneck in the whole system} if it is dominantly requested by eligible users at each server.
For the resulting allocation of Problem~\ref{P1}, we may show bottleneck fairness in a more general \emph{sense}.
Specifically, assume that there exists one resource $\rho(i)$ at each server $i$ for which the inequality in $\eqref{BF_criteria}$ is satisfied when substituting $\rho$ with $\rho(i)$. We refer to $\rho(i)$ as the \emph{bottleneck resource at server $i$}.

\begin{theorem}\label{Th_Universal_Allocation}
Assume there exists a bottleneck resource at each server.
An allocation, ${\bf x}$, is a solution to Problem~\ref{P1} if and only if it satisfies PS-DSF.
\end{theorem}

The proof appears in the appendix.
Under the conditions in Theorem~\ref{Th_Universal_Allocation}, different notions of fairness (including different variants of $\alpha$PF-VDS for $g_i(z)=z^{-\alpha},~\alpha>0$) coincide with PS-DSF. Hence, the resulting allocation of Problem~\ref{P1} satisfies max-min fairness with respect to the bottleneck resource at each server~\cite{PSDSF}. The following corollary follows directly from the proof of Theorem~\ref{Th_Universal_Allocation}.


\begin{corollary}\label{Coro_universal}
If there exists a bottleneck resource at server $i$,
then the resulting allocation of Problem~\ref{P1} satisfies max-min fairness with respect to the bottleneck resource at this server.
\end{corollary}

\subsection{Extensions}\label{sec:main:ext}
\noindent{\bf Non-divisible servers:}
Problem~\ref{P1} is formulated based on the assumption that the available resources over each server are arbitrarily divisible. For a data-center comprising of a plurality of servers, it is sometimes of practical interest to assume that servers
may not be divided to finer partitions~\cite{Ghodsi13}, so that each server may only be time-shared by different users. In this case, let $x_{n,i}$ denote the average number of tasks that are executed by server $i$ for user $n$ per unit of time. Accordingly, $x_{n,i}/\gamma_{n,i}$ gives the percentage of time unit that server $i$ is allocated to user $n$. Hence, we have the following condition on an allocation, ${\bf x}$, to be feasible~\cite{PSDSF}:
\be
\sum_{n\in\mathcal{N}_i}\frac{x_{n,i}}{\gamma_{n,i}}\le 1,~\forall i.
\ee
\begin{theorem}\label{Th_non_divisible_servers}
Let $g_i(z)$ be a continuously differentiable, strictly concave and increasing function. An allocation, ${\bf x}$,
satisfies PS-DSF if and only if it is a solution to the following problem.
\begin{problem}\label{P2} For every server $i$:
\be
&& \max_{{\bf x}_i }\sum_{n\in\mathcal{N}_i}\phi_ng_i\left(\frac{x_n}{\phi_n\gamma_{n,i}}\right)\label{P2_1}\qquad\qquad\\
&& \text{Subject to:}~\sum_{n\in\mathcal{N}_i}\frac{x_{n,i}}{\gamma_{n,i}}\le 1,\label{P2_2}\\
&& \text{~~~~~~~~~~~~~~}x_{n,i}\ge0,~\forall n\in\mathcal{N}_i,\label{P2_3}\\
&& \text{~~~~~~~~~~~~~~}x_{n,i}=0,~\forall n\notin\mathcal{N}_i.\label{P2_4}
\ee
\end{problem}
\end{theorem}

Theorem~\ref{Th_non_divisible_servers} implies that there exists an \emph{ideal} allocation for non-divisible servers, which results from Problem~\ref{P2} for any arbitrary $g_i(z)$, provided that $g_i(z)$ is continuously differentiable, strictly concave, and increasing. In fact, different notions of fairness (including different variants of $\alpha$PF-VDS, PS-DSF ($\alpha\rightarrow\infty$) and proportional fairness ($\alpha=1$))
coincide in this case. The major advantage of the PS-DSF allocation mechanism (or $\alpha$PF-VDS in general) is giving a simple per-server criterion to find such an ideal allocation~\cite{PSDSF}.

\noindent{\bf Servers with heterogeneous objectives:}
In Section~\ref{sec:main:Def}, we chose $g_i(z)$ from the class of $\alpha$-fair utility functions
with $g'_i(z)=z^{-\alpha}$, for some $\alpha$ that is fixed for all servers.
In general, different objectives may be followed for resource allocation at different servers.
For instance, in a cloud computing environment, some servers may be owned by a group of users. For such proprietary servers, there might be more emphasis on \emph{fair} allocation of the resources among proprietors (eligible users). On the other hand, there might be some public servers
for which there might be more emphasis on \emph{efficient} allocation of the resources.

To address this issue, one may choose $g_i(z)$ from the class of $\alpha$-fair utility functions, but (possibly) with different $\alpha_i$ for different servers. In this case, we may extend all the results shown in Theorem~\ref{Th_basic}-\ref{Th_Properties}. In particular, we may extend the definition of $\alpha$PF-VDS, so that the resulting allocation of Problem~\ref{P1} satisfies the inequality in \eqref{fair_criteria}, where we consider different $\alpha_i$ for different servers. Also, it can be observed that the resulting allocation of Problem~\ref{P1} satisfies PS-DSF with respect to each server $i$ as $\alpha_i\rightarrow\infty$ (see proof of Theorem~\ref{Th_psdsf} in the appendix). Furthermore, careful inspection of the proof of Theorem~\ref{Th_Properties} indicates that the resulting allocation of Problem~\ref{P1} in such a heterogeneous setting satisfies sharing incentive and envy freeness properties, provided that $\alpha_i\ge1, ~\forall i$.

Finally, we introduce a broader class of utility functions which includes the class of $\alpha$-fair utility functions.
Specifically, consider $g_i(z)$ such that:
\be
g'_i(z)=\left(\frac{A_i}{z}\right)^{\alpha_i}+\frac{B_i}{z},\label{Extended_utility}
\ee
where $A_i,B_i\ge0$, and $\alpha_i>0$.
For this class of utility functions, we can show uniqueness of the solution to Problem~\ref{P1} (in terms of $\{x_n\}$), when $B_i>0$ (see Section~\ref{sec:reformulate:structure}).
Again, the trade-off between efficiency and fairness can be adjusted by $\alpha_i$.
Furthermore, we may extend Theorem~\ref{Th_Properties} to show envy freeness and sharing incentive properties for this class of utility functions. 

\begin{corollary}\label{Corollary_EF_SI}
The resulting allocation of Problem~\ref{P1} satisfies sharing incentive and envy-freeness properties,
provided that $g_i(z)$ is from the class of functions in \eqref{Extended_utility}, and $\alpha_i\ge1$.
\end{corollary}
\begin{corollary}\label{Corollary_PSDSF}
Let $g_i(z)$ be from the class of utility functions specified by \eqref{Extended_utility}.
The resulting allocation of Problem~\ref{P1} satisfies PS-DSF with respect to server $i$ as $\alpha_i\rightarrow\infty$,
provided that $\tilde{s}_{n,i}/A_i\le1,~\forall n$.
\end{corollary}
The proofs appear in the appendix.

\section{Towards a solution to Problem~\ref{P1}: an equivalent formulation}\label{sec:reformulate}
In this section, we reformulate Problem~\ref{P1} as a nonlinear complementary problem.
This equivalent formulation forms the basis to develop an iterative algorithm to solve this problem in the next section.

\subsection{Formulation as a non-linear complementary problem}
For Problem~\ref{P1} describing a concave game, it is well known that ${\bf x}$ is a solution (Nash equilibrium) if and only if there exists a set of multipliers, ${\bf \lambda}$ and ${\bf \nu}$, such that KKT conditions are satisfied\footnote{Conditions of the form $0\le x\perp y\ge0$ are referred to as \emph{complementary conditions}, where $x\perp y$ means that $xy=0$.}~\cite{GNE_KKT}:
\be
&& 0\le\lambda_{i,r} \perp (c_{i,r}-\sum_{n\in\mathcal{N}_i}x_{n,i}d_{n,r})\ge0,~\forall r,i,\label{KKT1_1}\qquad\qquad\\
&& 0\le\nu_{n,i} \perp x_{n,i}\ge0,~n\in\mathcal{N}_i,~\forall i,\label{KKT1_2}\\
&&  \frac{\partial\mathcal{L}_i({\bf x},{\bf \lambda},{\bf\nu})}{\partial x_{n,i}}=0,~n\in\mathcal{N}_i,~\forall i,\label{KKT1_3}
\ee
where, $\mathcal{L}_i({\bf x},{\bf \lambda},{\bf\nu})$ is the Lagrangian function for the local problem at server $i$,
\be\nonumber
\mathcal{L}_i({\bf x},{\bf \lambda},{\bf\nu}) := \sum_{n\in\mathcal{N}_i}\left[\phi_ng_i\left(\frac{x_n}{\phi_n\gamma_{n,i}}\right)+\nu_{n,i}x_{n,i}\right]\\
  +~ \sum_r\lambda_{i,r}(c_{i,r}-\sum_{n\in\mathcal{N}_i}x_{n,i}d_{n,r}).\qquad
\ee
Hence, the first order optimality condition in \eqref{KKT1_3} implies that:
\be
\frac{1}{\gamma_{n,i}}g'_i\left(\frac{x_n}{\phi_n\gamma_{n,i}}\right)-\sum_r\lambda_{i,r}d_{n,r}+\nu_{n,i}=0, ~n\in\mathcal{N}_i,~\forall i.\label{FOC_1}
\ee
We may solve the system of KKT conditions in \eqref{KKT1_1}-\eqref{KKT1_3} for $\nu_{n,i}$,
and reach the following simplified set of conditions:
\be
&& 0\le\lambda_{i,r} \perp f_{i,r}({\bf x}_i)\ge0,~\forall r,i,\label{KKT2_1}\qquad\qquad\qquad\\
&& 0\le x_{n,i} \perp f_{n,i}({\bf x},{\bf \lambda}_i)\ge0,~n\in\mathcal{N}_i,~\forall i,\label{KKT2_2}
\ee
where,
\be
&& f_{i,r}({\bf x}_i) := c_{i,r}-\sum_{n\in\mathcal{N}_i}x_{n,i}d_{n,r},\label{mapping1}\\
&& f_{n,i}({\bf x},{\bf \lambda}_i):=
\sum_r\lambda_{i,r}d_{n,r}-\frac{1}{\gamma_{n,i}}g'_i\left(\frac{x_n}{\phi_n\gamma_{n,i}}\right).\qquad\label{mapping2}
\ee

The problem of finding $({\bf x},{\bf \lambda})$, such that the complementary conditions in \eqref{KKT2_1}-\eqref{KKT2_2} are satisfied,
is known as a Non-linear Complementary Problem (NCP).
A brief introduction to this family of problems is given in the next subsection.


\subsection{Background on nonlinear complementary problems}
Let $\mathcal{F}({\bf z})$ be a continuously differentiable function (mapping), $\mathcal{F}:\mathbb{R}^m\rightarrow\mathbb{R}^m$. The nonlinear complementary problem, NCP$(\mathcal{F})$,
is to find ${\bf z}\in{\mathbb{R}^m}$ such that:
\be
0\le {\bf z}, \mathcal{F}({\bf z})\ge0, {\bf z}^T\mathcal{F}({\bf z})=0,
\ee
where the inequalities are taken componentwise.
This problem can be best described by introducing a \emph{complementary function}.
Specifically, $\psi:{\mathbb R}^2\rightarrow\mathbb{R}$ is said to be a \emph{complementary function} when~\cite{Fischer98}:
\be
\psi(a,b)=0 ~\Leftrightarrow~ a\ge0,~b\ge0,\text{ and }ab=0.\label{Criteria_NCP_function}
\ee
Based on a complementary function, NCP$(\mathcal{F})$ can be reformulated as:
\be
\psi(z_l,F_l({\bf z}))=0,~\forall l.
\ee

Several complementary functions have been proposed in the literature~\cite{ResolutionNCP}, but the most prominent one is the Fischer-Burmeister function~\cite{Fischer98}:
\be
\psi_{FB}(a,b) := \sqrt{a^2+b^2}-a-b.
\ee
Here we use the following complementary function~\cite{ResolutionNCP}:
\be
\psi(a,b)=\frac{1}{2}\psi_{FB}^2(a,b),\label{NCP_function}
\ee
which is continuously differentiable, and $\psi(a,b)\ge0$, $\forall(a,b)\in\mathbb{R}^2$.

%
%

\subsection{A constrained merit function}
In order to solve the NCP described by \eqref{KKT2_1}-\eqref{KKT2_2}, 
we confine the feasible region such that \eqref{KKT2_1} is always satisfied. 
\begin{lemma}\label{Corollary_merit_function}
$({\bf x},{\lambda})$ is a solution to the NCP described by \eqref{KKT2_1}-\eqref{KKT2_2}, if and only if it is a solution to the following problem.
\begin{problem}\label{P3}
\be
&&\min\Psi({\bf x},{\bf \lambda}):=\sum_{i\in\mathcal{K}}\sum_{n\in\mathcal{N}_i}
\psi (x_{n,i},f_{n,i}({\bf x},{\bf \lambda}_i)),\quad\:\:\label{P3_0}\\
&&\text{s.t.}\quad \sum_{n\in\mathcal{N}_i}x_{n,i}d_{n,r}\le c_{i,r},~\forall r,i,\label{P3_1}\\
&&\qquad\:\,       \lambda_{i,r}\ge0,~r\in \mathcal{R}_i({\bf x}),\forall i,\label{P3_2} \\
&&\qquad\:\,       \lambda_{i,r} = 0,~r\notin \mathcal{R}_i({\bf x}),\forall i.\label{P3_3}
\ee
where, $\mathcal{R}_i({\bf x}):=\{r\mid \sum_{n\in\mathcal{N}_i}x_{n,i}d_{n,r}=c_{i,r}\}$ denotes the set of \emph{saturated resources} at server $i$ under the allocation ${\bf x}$.
\end{problem}
\end{lemma}
\begin{proof}
The constraints in \eqref{P3_1}-\eqref{P3_3} are established if and only if $({\bf x},\lambda)$ satisfies \eqref{KKT2_1}. According to \eqref{Criteria_NCP_function}, the complementary conditions in \eqref{KKT2_2} are satisfied if and only if $\psi (x_{n,i},f_{n,i}({\bf x},{\bf \lambda}_i))=0$, $n\in\mathcal{N}_i,~\forall i$. On the other hand \eqref{NCP_function} implies that $\psi(\cdot,\cdot)$ has a lower bound of zero. Hence, a feasible point $({\bf x},\lambda)$ is a solution to Problem~\ref{P3} if and only if \eqref{KKT2_2} is satisfied for all $n\in\mathcal{N}_i,~\forall i$.
\end{proof}

The following theorem, proven in the appendix, suggests how to find global minima for Problem~\ref{P3}.
\begin{theorem}\label{Th_stationary_points}
A feasible point, $({\bf x},{\bf \lambda})$, satisfying the conditions  in \eqref{P3_1}-\eqref{P3_3}, is a solution to Problem~\ref{P3}
if and only if it is a stationary point of $\Psi({\bf x},\lambda)$ in \eqref{P3_0}, i.e., $\nabla\Psi({\bf x},\lambda)=0$.
\end{theorem}

\subsection{Structure of the solution}\label{sec:reformulate:structure}
In this section, we investigate the structure of the solution set for Problem~\ref{P1}, or the equivalent NCP described by Problem~\ref{P3},
when $g_i(\cdot)$ is specified by \eqref{Extended_utility} and $\alpha_i\ge1,~\forall i$. In Problem~\ref{P1}, the feasible region for ${\bf x}$, described by \eqref{P1_2}-\eqref{P1_4}, is a bounded region. As we showed in Section~\ref{sec:main:properties}, the solution to this problem satisfies sharing incentive property. That is,
\be
x_n\ge\frac{\psi_n}{\sum_m\phi_m}\sum_i\gamma_{n,i}=:x_n^{\min},~\forall n.\label{lower_bound_tasks}
\ee
For the equivalent NCP described by \eqref{KKT2_1}-\eqref{KKT2_2}, it can also be shown that the solution $({\bf x},{\bf \lambda})$ is contained within a bounded region. To observe boundedness of Lagrange multipliers for each server $i$, consider some user $n$ for which $x_{n,i}d_{n,r}>0$, for some resource $r$. From the complementary condition in~\eqref{KKT2_2}, it follows that $f_{n,i}({\bf x},{\lambda}_i)=0$, and therefore,
\be
\sum_r\lambda_{i,r}d_{n,r} & = &  \frac{1}{\gamma_{n,i}}g'_i\left(\frac{x_n}{\phi_n\gamma_{n,i}}\right).
\ee
Since $g_i(\cdot)$ is a concave function, the lower bound in \eqref{lower_bound_tasks} implies that:
\be
\sum_r\lambda_{i,r}d_{n,r} &\le& \frac{1}{\gamma_{n,i}}g'_i\left(\frac{x_n^{\min}}{\phi_n\gamma_{n,i}}\right)<\infty,
\ee
which clearly shows boundedness of Lagrange multipliers for server $i$.

\begin{lemma}\label{lem_connectedness}
The solution set for the NCP described by \eqref{KKT2_1}-\eqref{KKT2_2}, or equivalently Problem~\ref{P3}, is \emph{connected}.
\end{lemma}
\begin{proof}
For Problem~\ref{P3} we know that:
(a) if $({\bf x},{\lambda})$ is a solution, then $\nabla\Psi({\bf x},{\lambda})=0$ (c.f. Theorem~\ref{Th_stationary_points}), (b) the function $\Psi({\bf x},{\lambda})$, $\Psi:\mathbb{R}^{(N+M)K}\mapsto \mathbb{R}$, is continuously differentiable, and (c) the solutions are contained within a bounded region. Given these conditions, the proofs for Lemma 3.1, and Corollary 3.5 of \cite{Facchinei98} can be applied to show connectedness of the solution set for this problem.
\end{proof}

The following result, borrowed from \cite{Facchinei98}, is a direct consequence of Lemma~\ref{lem_connectedness}.
\begin{corollary}\label{corollary_local_global1}
Problem~\ref{P3} has a unique solution if and only if it has a locally unique solution \cite{Facchinei98}.
\end{corollary}

In the same way, it can be observed that Problem~\ref{P3} has a unique solution in terms of the total allocated tasks to each user, provided that it has a locally unique solution in terms of $\{x_n\}$. In the following we draw conditions under which the optimal solution to Problem~\ref{P3} will be locally unique in terms of $\{x_n\}$.


\begin{theorem}\label{Th_uniqueness}
Let $({\bf x},{\bf \lambda})$ denote an optimal solution to Problem~\ref{P3},
when $g_i(\cdot)$ is from the class of functions specified in \eqref{Extended_utility},
and $\alpha_i\ge 1$, $A_i\ge0$ and $B_i>0$. 
If $\lambda_{i,r}>0$ only for one resource at each server $i$, then $({\bf x},{\bf \lambda})$
will be locally unique in terms of the total allocated tasks to each user, $\{x_n\}$.
\end{theorem}

\section{Iterative solution and distributed implementation}\label{sec:algo}
In this section, we first develop an iterative (centralized) algorithm that is globally convergent to an optimal solution
(Nash equilibrium) to Problem~\ref{P1}. 
Next, we will propose a simple heuristic to solve this problem in a distributed fashion.



\subsection{Centralized solution}\label{sec:algo:ctr}
As discussed in Section~\ref{sec:reformulate}, ${\bf x}$ is a solution to Problem~\ref{P1},
if and only if there exists a set of multipliers such that $({\bf x},{\bf \lambda})$ is a solution to Problem~\ref{P3}.
According to Theorem~\ref{Th_stationary_points}, in order to find a solution to Problem~\ref{P3}, we may employ an iterative descent algorithm which converges to a stationary point where $\nabla\Psi({\bf x},{\bf \lambda})=0$. In the following, we propose an iterative algorithm inspired by projected-gradient method.

Initially, we begin with a feasible point, $({\bf x}^1,{\bf \lambda}^1)$, which satisfies \eqref{P3_1}-\eqref{P3_3}.
Then, in each iteration $h\ge1$, we update $({\bf x}^h, {\bf \lambda}^h)$ such that $\Psi({\bf x}^{h+1},{\bf \lambda}^{h+1})$ is decreased compared to $\Psi({\bf x}^{h},{\bf \lambda}^{h})$, while $({\bf x}^{h+1},{\bf \lambda}^{h+1})$ remains feasible. Specifically, let $\mathcal{R}_i^h$ denote
the set of saturated resources at server $i$ under the allocation ${\bf x}^h$,
\be
\mathcal{R}_i^h:=\{r\mid \sum_{n\in\mathcal{N}_i}x_{n,i}^hd_{n,r} = c_{i,r}\}.
\ee
The opposite of gradient, $-\nabla\Psi({\bf x}^{h},{\bf \lambda}^{h})$, is a descent direction.
However, by moving ${\bf x}^h$ in the direction of $-\nabla_{\bf x}\Psi({\bf x}^{h},{\bf \lambda}^{h})$,
the capacity constraint would be violated for resource $r\in\mathcal{R}_i^h$ if:
\be
\beta_{i,r}^h:=-(\nabla_{{\bf x}_i}\Psi({\bf x}^{h},{\bf \lambda}^{h}))^T{{\bf d}_{r}}>0,
\ee
where ${\bf d}_{r}:=[d_{n,r}]_{N\times1}$, and $\nabla_{{\bf x}_i}\Psi$ is the gradient of $\Psi$ with respect to ${\bf x}_i$. So, the moving direction should be chosen in a way that the capacity constraints are not violated for any resource $r\in\mathcal{R}_i^h$.
Furthermore, to reach a feasible point satisfying \eqref{P3_3}, we need to preserve equality for the capacity constraints 
corresponding to resources $r\in\mathcal{R}_i^h$ with $\lambda_{i,r}^h>0$. Hence, we choose the moving direction for updating ${\bf x}^h_i$, $({\bf v}_x^h)_i$, as the projection of
$-\nabla_{{\bf x}_i}\Psi({\bf x}^{h},{\bf \lambda}^{h})$ onto:
\be\nonumber
\Omega_i:=\{{\bf v}\in\mathbb{R}^N  \mid {\bf v}^T{\bf d}_{r}\le0,~ r\in\mathcal{R}_i^h:~\lambda_{i,r}=0,\:\\
                      \text{~and~~}{\bf v}^T{\bf d}_{r}=0,~ r\in\mathcal{R}_i^h:~\lambda_{i,r}>0 \}.\label{projection_set}
\ee
To update ${\bf \lambda}_{i,r}^h$, the following is a descent direction\footnote{The function $\text{U}(\cdot)$ is $\text{U}(z)=1$ if $z\ge0$, and $\text{U}(z)=0$ otherwise.}:
\be
(\tilde{\bf v}^h_{\lambda})_{i,r}:=-\frac{\partial\Psi}{\partial\lambda_{i,r}}
+ \beta^h_{i,r}~\text{U}\left(-\beta^h_{i,r}\frac{\partial\Psi}{\partial\lambda_{i,r}}\right).\label{lambda_non_feas_dir}
\ee
However, to maintain feasibility, we need to choose the moving direction $({\bf v}^h_\lambda)_{i,r}$ such that:
\be
({\bf v}^h_\lambda)_{i,r}:=[(\tilde{\bf v}^h_\lambda)_{i,r}]^+_{\lambda_{i,r}},\label{lambda_feas_dir}
\ee
where, $[v]^+_{z}=v$ if $z>0$, and otherwise $[v]^+_{z}=\max\{v,0\}$.
Finally, we update
\be
&& {\bf x}^{h+1}={\bf x}^{h}+\eta^h{\bf v}_x^h,\\
&& {\bf \lambda}^{h+1}={\bf \lambda}^{h}+\eta^h{\bf v}_\lambda^h,\qquad\qquad\qquad\;
\ee
where, the step size $\eta^h$ is chosen such that:
\be
&& \Psi({\bf x}^{h+1},{\lambda}^{h+1}) - \Psi({\bf x}^{h},{\lambda}^{h}) < 0,\qquad\label{Step_size_C1}\\
&& \sum_{n\in\mathcal{N}_i}x_{n,i}^{h+1}d_{n,r}\le c_{i,r},~\forall r,i,\label{Step_size_C2}\\
&& \lambda_{i,r}^{h+1}\ge 0,~\forall r,i.\label{Step_size_C4}
\ee
The algorithm terminates when $\|{\bf v}^h\|<\epsilon$, for some arbitrary $\epsilon>0$.
The above described algorithm, referred to as \emph{Per-Server Multi-resource Fair Allocation} (PS-MFA) algorithm has been summarized in Algorithm I.

\begin{table}
\footnotesize
\captionsetup{labelformat=empty}
\caption{Algorithm I: {PS-MFA Algorithm}}
\label{table2}
\begin{tabular}{p{8.48cm}}
\hline\noalign{\smallskip}
Initially, begin with a feasible point $({\bf x^1},{\bf \lambda^1})$ satisfying \eqref{P3_1}-\eqref{P3_3}. Then, in each iteration $h\ge1$, take the following steps.

\begin{itemize}
\item [1)] Choose the moving direction for ${\bf x}_i$, as the projection of
$-\nabla_{{\bf x}_i}\Psi({\bf x}^h,\lambda^h)$ onto $\Omega_i$,
\be
({\bf v}_x^h)_i=\text{Proj}_{\Omega_i}(-\nabla_{{\bf x}_i}\Psi({\bf x}^h,\lambda^h)),
\ee
where $\Omega_i$ is given by~\eqref{projection_set}.
\item [2)] Choose the moving direction for $\lambda_{i,r}$ according to \eqref{lambda_feas_dir}.
\item [3)] Choose the step size, $\eta^h$, sufficiently small, so that the conditions in \eqref{Step_size_C1}-\eqref{Step_size_C4} are satisfied.
\item [4)] Let ${\bf x}^{h+1}={\bf x}^{h}+\eta^h{\bf v}_x^h$,~
           and ${\bf \lambda}^{h+1}={\bf \lambda}^{h}+\eta^h{\bf v}_{\lambda}^h$.
\item [5)] Stop when $\|{\bf v}^h\|\le\varepsilon$, where ${\bf v}^h = [{\bf v}_x^h;{\bf v}_{\lambda}^h]$.
\end{itemize}\\
\noalign{\smallskip}\hline
\end{tabular}
\end{table}

\begin{lemma}\label{lem_descent}
The moving direction, ${\bf v}^h=[{\bf v}_x^h;{\bf v}_\lambda^h]$, is a strictly descent direction, unless $\|{\bf v}^h\|=0$.
\end{lemma}

The proof appears in the appendix. To further analyze the convergence point of the PS-MFA algorithm,
we need to make one of the following assumptions.

\begin{assumption}\label{assmp_degen}
The PS-MFA algorithm terminates at a point where $\lambda_{i,r}>0$ for every $r\in\mathcal{R}_i,~\forall i$.
\end{assumption}

\begin{assumption}\label{assmp_simple}
For every server $i$ and resource $r$ there exists at least one user $n\notin\mathcal{N}_i$ with $d_{n,r}>0$.
\end{assumption}

Assumption~\ref{assmp_degen} requires that the algorithm terminates at a \emph{non-degenerate point},
where $\lambda_{i,r}>0$ for every $r\in\mathcal{R}_i$.
Assumption~\ref{assmp_simple} does not restrict the solution,
but it gently restricts the model. In particular, it holds when all users demand all types of resources,
and there exists one user $n\notin\mathcal{N}_i$ for each server $i$.

\begin{lemma}\label{lem_eqv}
Under either of Assumptions~\ref{assmp_degen} or \ref{assmp_simple}, ${\bf v}^h=0$ only when $\nabla_{\bf x}\Psi=0$ (the proof is given in the appendix).
\end{lemma}

\begin{theorem}\label{Th_convergence}
The PS-MFA algorithm terminates at a stationary point of $\Psi$ under either of Assumptions~\ref{assmp_degen} or \ref{assmp_simple}. 
\end{theorem}

\begin{proof}
The facts that $\Psi$ is lower bounded and ${\bf v}^h$ is a descent direction imply that the algorithm converges/terminates.
When the algorithm terminates, $\|{\bf v}^h\|\rightarrow0$, and therefore $\nabla_{\bf x}\Psi=0$ (see Lemma~\ref{lem_eqv}).
This in turn implies that $\psi (x_{n,i},f_{n,i}({\bf x},{\bf \lambda}_i))=0,~\forall n,i$ (see the proof of Theorem~\ref{Th_stationary_points}), and therefore $\nabla\Psi=0$.
\end{proof}

According to Theorem~\ref{Th_stationary_points} and~\ref{Th_convergence}, the sequence of allocations $\{{\bf x}^h\}$ generated by the PS-MFA algorithm globally converges to an optimal solution to Problem~\ref{P1} under either of Assumption~\ref{assmp_degen} or \ref{assmp_simple}.

\subsection{Distributed implementation}\label{sec:algo:dtr}
The PS-MFA algorithm, as described by Algorithm~I, may run in parallel on different servers,
where each server iteratively updates its own allocation parameters, $({\bf x}_i,\lambda_i)$.
However, to find the gradient vector at each server $i$, one needs to know the allocation parameters at the other servers.
To achieve an efficient distributed implementation, we propose a simple heuristic algorithm which directly applies to the NCP in \eqref{KKT2_1}-\eqref{KKT2_2}. First, assume that the Lagrange multipliers, $\{\lambda_{i,r}\}$, are known,
so that the complementary conditions in \eqref{KKT2_1} are satisfied. Starting with a feasible allocation ${\bf x}^1$, in each iteration $h\ge1$,
one may update $x^h_{n,i}$ in the following direction: 
\be
v^{h}_{n,i} &=& \left[-f_{n,i}({\bf x}^h, \lambda)\right]_{x^h_{n,i}}^+,~n\in\mathcal{N}_i,~\forall i,\label{dist_impl_1}\\
x_{n,i}^{h+1} &=& \left[x_{n,i}^h + \kappa_iv^{h}_{n,i}\right]^+,\label{dist_impl_2}
\ee
where, $[v]_z^+=v$ if $z>0$, and otherwise $[v]_z^+=\max\{v,0\}$.
The step size, $\kappa_i>0$, is chosen independently by each server.
According to \eqref{dist_impl_1} and \eqref{dist_impl_2}, $x^h_{n,i}$ is decreased if it is positive and $f_{n,i}({\bf x}^h,\lambda_i)>0$.
Otherwise, it will be increased when $f_{n,i}({\bf x}^h,\lambda_i)<0$.
The above dynamic converges/terminates when the complementary conditions in \eqref{KKT2_2} are satisfied.
The only issue is to find the Lagrange multipliers. As also can be seen in \cite{kelly1998rate},
the Lagrange multipliers could be approximated by:
\be
\lambda_{i,r}({\bf x}^h)=\left[\sum_mx^h_{m,i}d_{m,r}-c_{i,r}+\varepsilon\right]^+/\varepsilon^2,
\ee
where a better approximation is achieved when $\varepsilon\rightarrow0$\footnote{This is equivalent to relaxing the capacity constraint in \eqref{P1_2} for every resource $r$, and adding a quadratic barrier function, $\left(\sum_mx_{m,i}d_{m,r}-c_{i,r}+\varepsilon\right)^2/2\varepsilon^2$, to the objective function in \eqref{P1_1}.}.
The major advantage of the above dynamic 
is that the allocation at each server could be updated
based on the local information on the available resources at each server, without any knowledge of the available resources at other servers. The only information that each server requires is the \emph{total} number of tasks that are allocated to eligible users, which is assumed to be updated whenever an update is made by any of the servers.
The convergence of such a heuristic algorithm is shown through numerical experiments in the next section.

\section{Trace driven simulation}\label{sec:eval}
In this section we evaluate the performance of the $\alpha$PF-VDS allocation mechanism through several experiments driven by real-world traces. Among the publicly available traces, Google cluster-usage data-set is the most extensive one which reports the resource usage for different tasks of different users (Google engineers and services) running on a cluster of 12000 servers over a period of one month.
The resource usage for each task has been measured at 1 second intervals, however, its average value is reported in the data-set with a period of 5 minutes. If a task is terminated during the five minutes period, its resource usage is reported over a shorter interval. Specifically, the data-set is given in a table format where each row reports the start and the end of each measurement period (which is typically 5 minutes), the job ID and task index, and finally the usage of CPU and RAM for the specified task.
The reader may refer to \cite{Google11} for further details.
Despite the detailed information that is provided by this data-set, it does not report the usage of other resources such as network/IO bandwidth.

To do experiments with resource demand vectors of higher dimensionality, we also use a data-set provided by
Bitbrain IT services incorporation which gives cloud service to users with business-critical workloads~\cite{shen2015}. This data-set reports the resource usage (of CPU, RAM, network and storage bandwidth) for different virtual machines each giving service to one user. 
Again, the measurements have been reported every 5 minutes for a period of 4 months.
There are only a few other traces which are publicly available (such as those from Yahoo or Facebook).
However, we do not use such traces herein as they do not provide the information that we need on the usage of different resources.
Furthermore, some of them pertain to specific processing tasks, such as Map-reduce, which may not represent the input to a real world data-center~\cite{shen2015}.

\subsection{Experimental setup}
To do experiments with Google traces, we consider a cluster consisting of four different classes of servers (120 servers in total as shown in Fig.~\ref{fig:num_exmp}), where the configuration of servers are drawn from the distribution of Google cluster servers~\cite{Google11}.
For the input workload, we randomly sample $2\%$ of users from the Google traces, so that the cluster is heavily loaded. The jobs for different users belong to 4 different scheduling classes (specified in a table for different jobs), where the last two classes are more latency sensitive. We classify users/jobs\footnote{While each user in practice may submit several jobs at the same time, for the sake of simplicity in presentation we assume that each user only submits one job. So, we may use jobs or users interchangeably.} into two different groups, $\mathcal{U}_1$ and $\mathcal{U}_2$, where the users in $\mathcal{U}_1$ are less latency sensitive.
Servers in classes A and B are assumed to be public (available to all users), while classes C and D are only available to delay sensitive users, $\mathcal{U}_2$.

%

\begin{figure}[h!]
\centering
\includegraphics[width = 0.99\columnwidth]{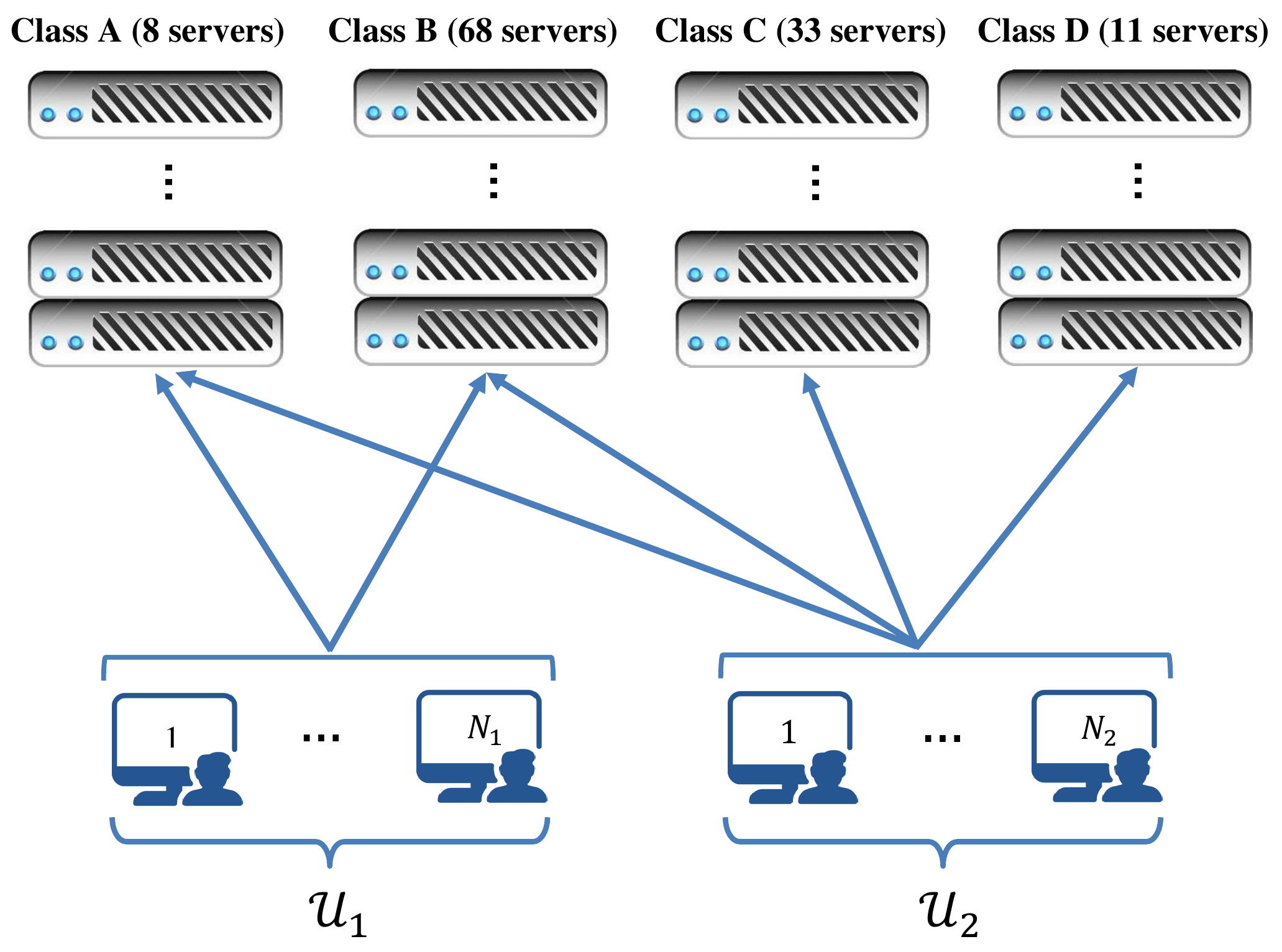}
\footnotesize
\caption{\footnotesize A cluster with four classes of servers (120 servers in total) and two classes of users.
Users are assumed to be equally weighted.
The configurations of resources (CPU and memory respectively) for servers of each class are as follows: $C_A=[1,1],~C_B=[0.5,0.5],~C_C=[0.5,0.25],~C_D=[0.5,0.75]$,
where CPU and memory units for each server are normalized with respect to the servers of the first class.}
\label{fig:num_exmp}
\vspace{-2mm}
\end{figure}

To do experiments with the Bitbrain data-set, first we need to choose a set of servers from which the resources are allocated to different users.
For this purpose, we find the overall resource usage by active users at different instants of time, so we provide enough capacity of each resource to meet the maximum overall usage. It is worth noting that users demand resources to meet their \emph{maximum} usage. Hence, the overall demands might be more than the overall resource usages at any of instant of time. Given the required resource capacities, we assume that CPU and RAM resources are provided by three types of servers, where there are 75 servers of type~1 with 4~GHz of CPU and 12~GBytes of RAM, 100~servers of type~2 with 8~GHz of CPU and 8~GBytes of RAM, and 75 servers of type~3 with 16~GHz of CPU and 4~GBytes of RAM. It is assumed that servers are distributed over 3 different locations as shown in Fig.~\ref{fig:num_exmp2}.
Each type of servers at each location is connected to a storage device and also is equipped with a network connection (see Fig.~\ref{fig:num_exmp2}). 
It is assumed that users are uniformly distributed over different locations.
Each user may get service from the servers at the same, or nearby location.

\begin{figure}[h!]
\centering
\includegraphics[width = 0.99\columnwidth]{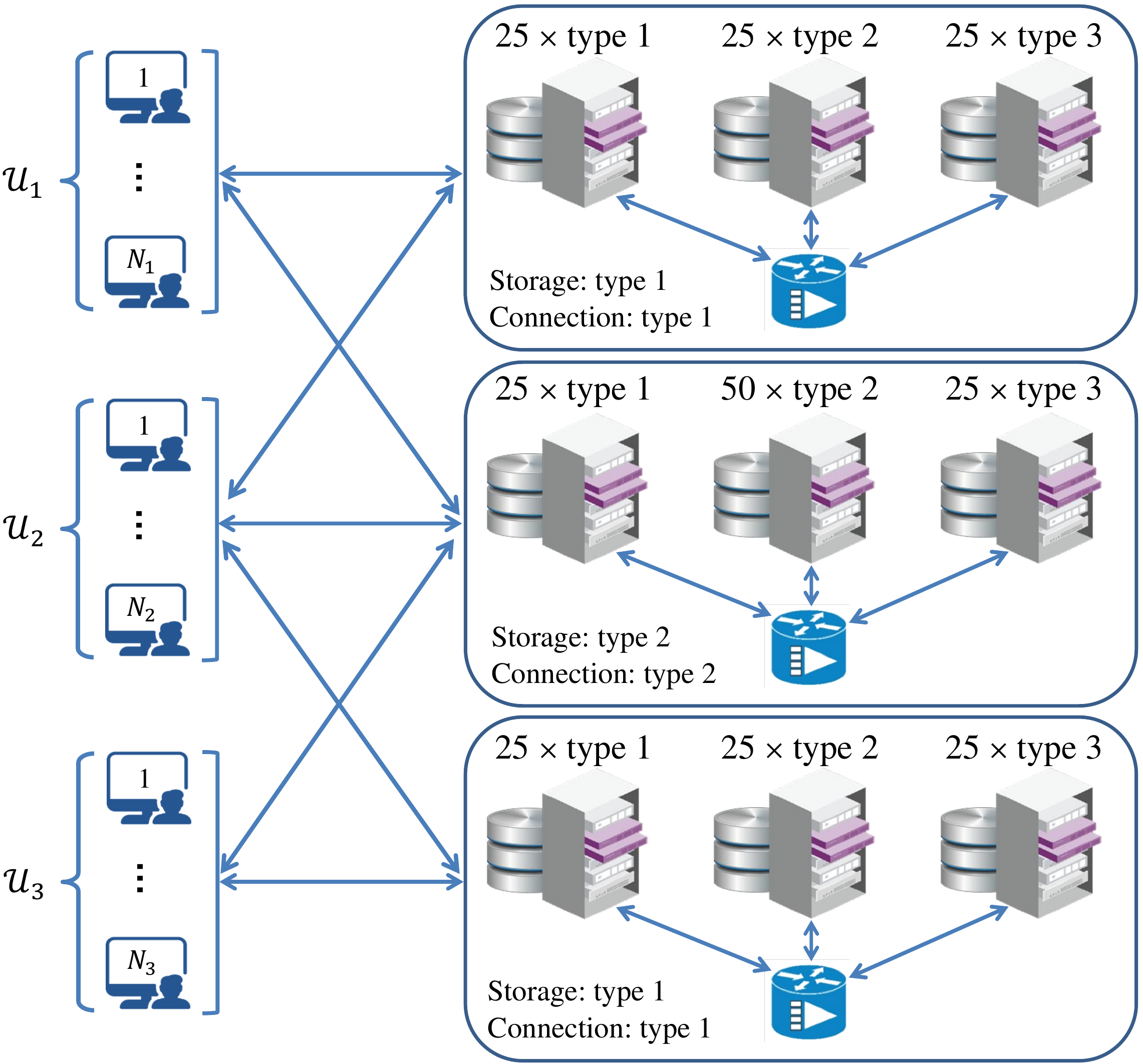}
\footnotesize
\caption{\footnotesize A data-center distributed over three different locations.
There are three types of servers, where the configuration of resources (CPU and memory respectively) for each type of server is as follows:
(4~GHz, 12~GBytes) for type~1, (8~GHz, 8~GBytes) for type~2, (16~GHz, 4~GBytes) for type~3. The storage devices are of two different types, where the read/write bandwidth for type~1, used at the first and the last locations, is 32~MB/s, and for type~2, used at the second location, is 100~MB/s. Finally, there are two types of broadband connections, where the first type provides a bandwidth of 100~Mb/s (and 1~Gb/s respectively) to send (receive) data, while the second type provides a bandwidth of 1~Gb/s (and 2~Gb/s respectively) to send (receive) data. 
}
\label{fig:num_exmp2}
\vspace{-2mm}
\end{figure}

\subsection{Adjusting the resource utilization}
As discussed in Section~\ref{sec:main:Def} and shown by an illustrative example in Fig.~\ref{fig:example3},
the resource utilization is improved as the parameter $\alpha$ in the $\alpha$PF-VDS allocation mechanism reduces.
In this subsection we study this effect when applying this mechanism to real-world workloads.

In the Bitbrain workload, users become active/incactive with a low churn.
So, the resources could be allocated to different users/virtual-machines in a semi-static manner.
In this case, we do several experiments for different sets of active users chosen at random instants of time.
In particular, for each set of active users we find the $\alpha$PF-VDS allocation for $\alpha=1$, $\alpha=3$ and $\alpha=\infty$.
In case of $\alpha=1$ and $\alpha=3$ we employ the distributed iterative algorithm proposed in Section~\ref{sec:algo:dtr}.
For $\alpha=\infty$, we employ the customized algorithm proposed in \cite{PSDSF} to implement PS-DSF.
Fig.~\ref{fig:cong} shows how our proposed distributed algorithm converges to an optimal solution to Problem~\ref{P3} where $\Psi({\bf x},\lambda)=0$. 
In Fig.~\ref{fig:diff_var_BB} we report the overall resource utilization that is achieved on average over different servers and over 100 runs, for different variants of $\alpha$PF-VDS. As expected, the $\alpha$PF-VDS results in a greater utilization of different resources for smaller values of $\alpha$. 
The improvement in utilization could be significant when $\alpha$ ranges from $\infty$ to 1.

\begin{figure}[h!]
\centering
\includegraphics[width = 2.5in]{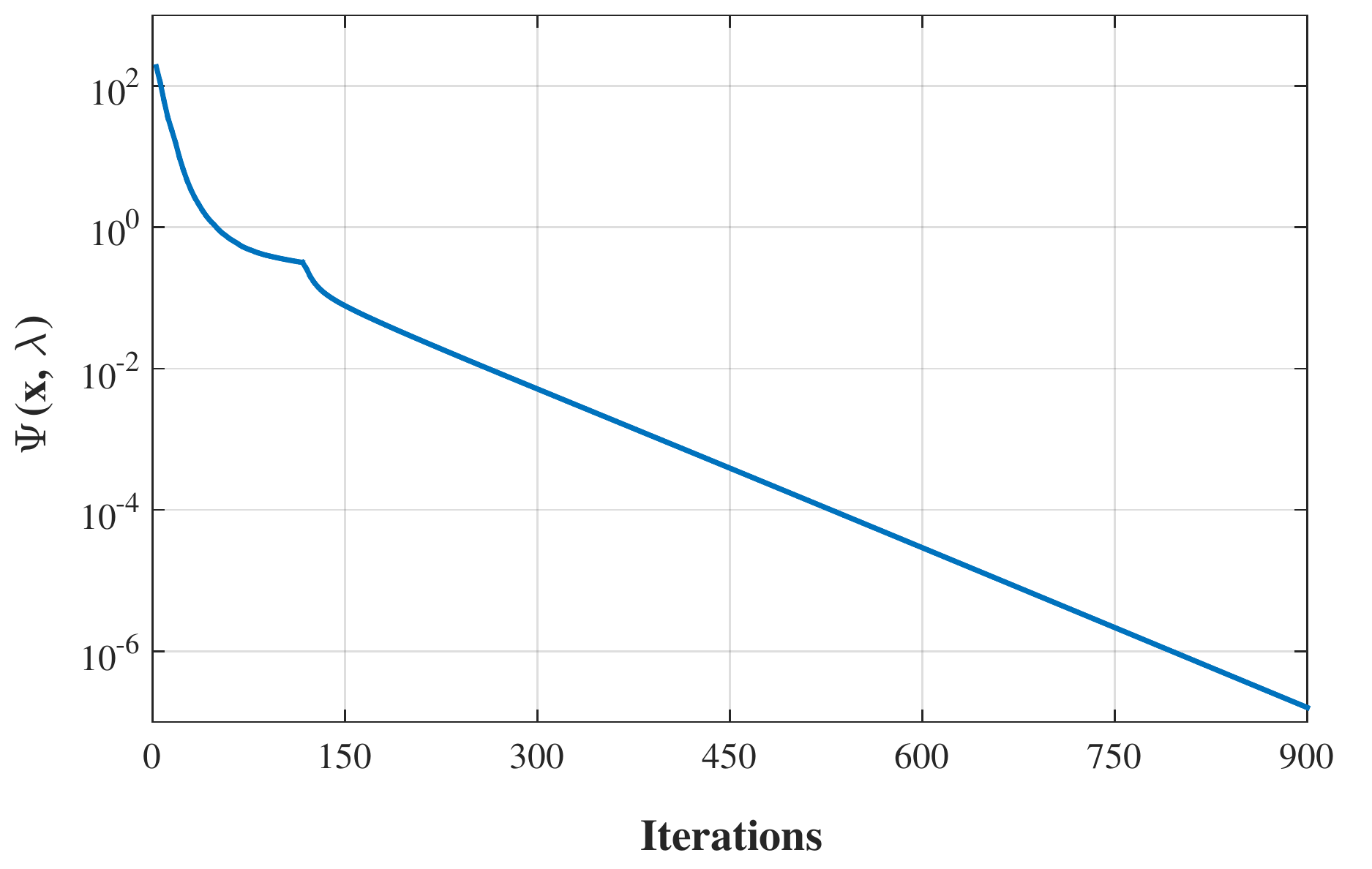}
\footnotesize
\caption{\footnotesize The convergence of the proposed distributed algorithm to find the $\alpha$PF-VDS allocation for the computing cluster in Fig.~\ref{fig:num_exmp2} when $\alpha=3$. The algorithm shows a linear convergence.}
\label{fig:cong}
\vspace{-2mm}
\end{figure}

\begin{figure}[h!]
\centering
\includegraphics[width = 0.95\columnwidth]{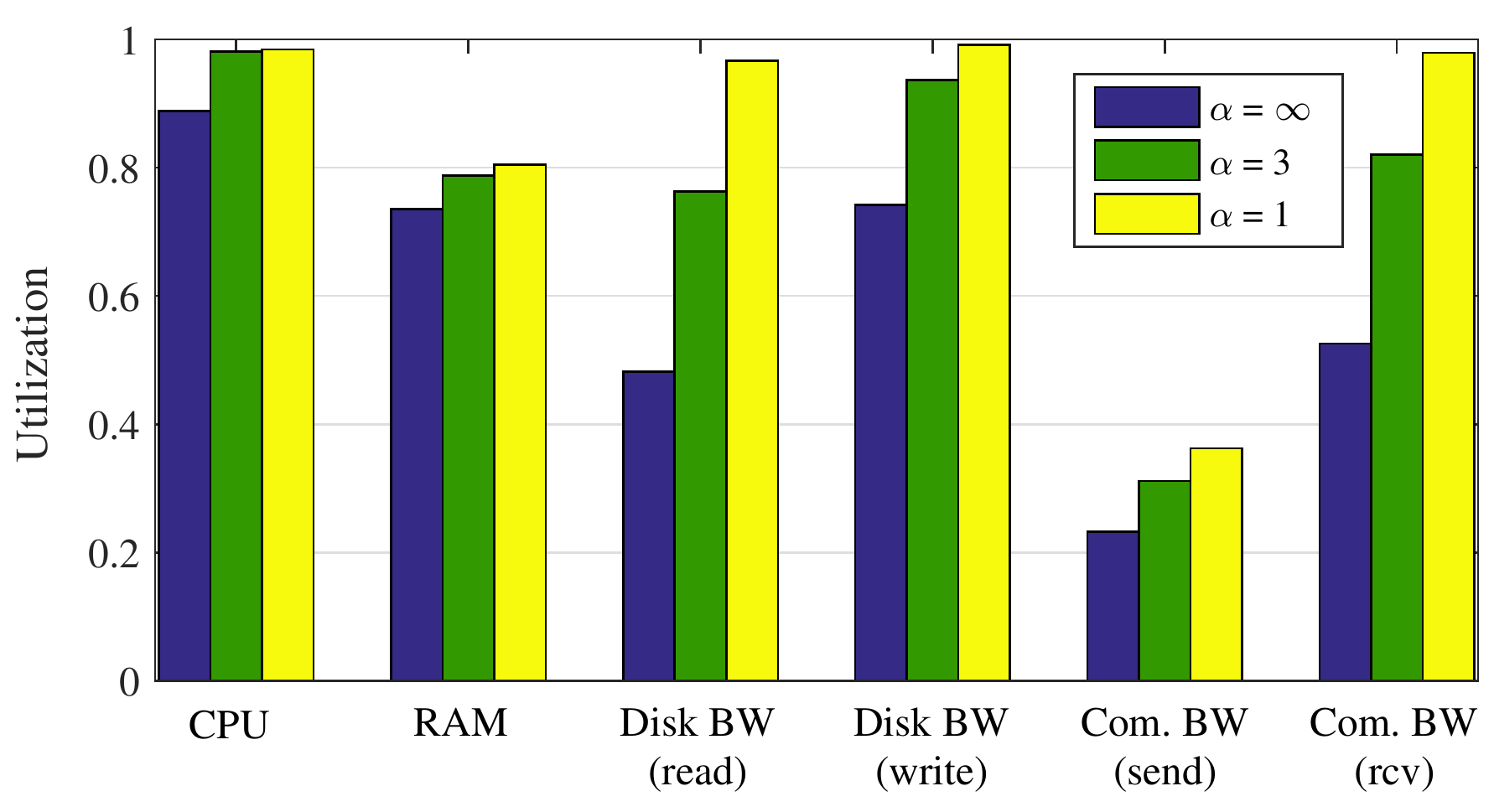}
\footnotesize
\caption{\footnotesize The overall resource utilization, averaged over different servers and over differen runs, for different variants of $\alpha$PF-VDS: $\alpha=1$ (proportional fairness), $\alpha=3$, $\alpha=\infty$ (PS-DSF).}
\label{fig:diff_var_BB}
\vspace{-2mm}
\end{figure}


For the Google workload, we allocate resources in a (semi)dynamic manner.
In particular, consider the computing cluster in Fig.~\ref{fig:num_exmp},
where $2\%$ of users from the Google traces are randomly chosen as the input workload.
In such a setting, we (re)allocate resources from the servers to demanding jobs (at least) every 5~minutes.
Specifically, given the resource usage for different tasks of each job by the Google traces, we may find the demand vector for each job (at the beginning of each 5 minutes interval) as the summation of the resource usage for different tasks (different tasks of the same job usually have proportional demands~\cite{Google11}). Given the total demand for each job, ${\bf d}_n=[d_{n,r}]$, we define a \emph{\bf quantum} for job $n$ as a block of resources in the amount of $\tilde{\bf d}_n:={\bf d}_n/\max_rd_{m,r}$ that is allocated to job $n$ for 1~second. Accordingly, job $n$ requires $q_n:= 300\max_rd_{m,r}$ execution quanta for the next 5~minutes interval. We use the normalized demand vectors, $\{\tilde{\bf d}_n\}$, as the input to the $\alpha$PF-VDS mechanism in order to find the number of tasks that are allocated to each job under this mechanism. Given the allocated tasks to each job, the completion time for job $n$ is given by $q_n/x_n$. If a job completes/leaves the system during the 5 minutes period, the released resources are reallocated among the remaining jobs.

The required execution quanta for different jobs, and also their activity duration, span a very wide range.
Fig.~\ref{fig:CDFJOB} plots the cumulative density function (CDF) of the activity duration (in terms of the number of 5 min intervals) and also the CDF of the required quanta for different jobs in the Google traces over an interval of 24 hours. As can be seen in Fig.~\ref{fig:CDFJOB}, around 38$\%$ of jobs are completed within a 5 min period, while 16$\%$ of them are active more than 24 hours. Moreover, 53$\%$ of jobs require less than 10 execution quanta, while $5\%$ of them require more then 10000~quanta.
The variety of jobs and the existence of such intensive ones indicate the necessity for an efficient and fair allocation mechanism,
which from one hand prevents intensive jobs starving others, and on the other hand results in an efficient resource utilization.

\begin{figure}[h!]
\centering
\includegraphics[width = 0.99\columnwidth]{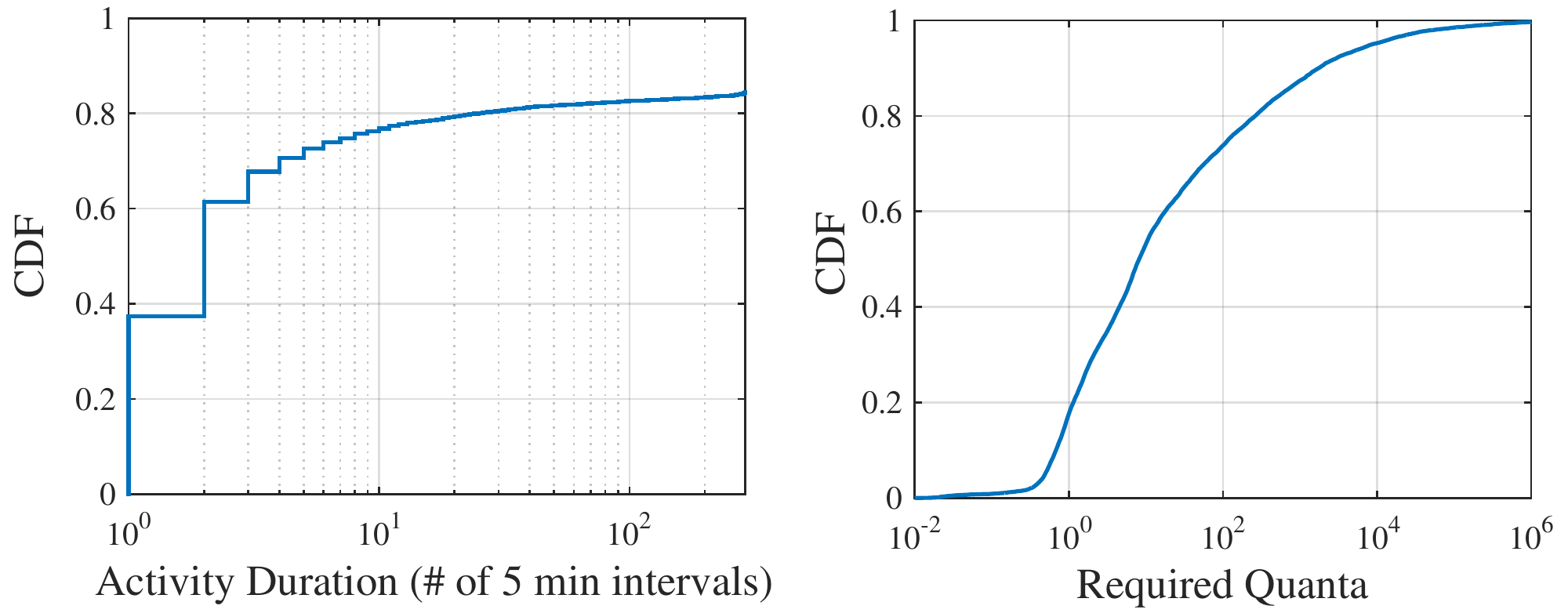}
\footnotesize
\caption{\footnotesize The CDF of the activity duration and the required quanta for different jobs over an interval of 24 hours.}
\label{fig:CDFJOB}
\vspace{-2mm}
\end{figure}

Fig.~\ref{fig:NumExmp:11} compares the overall resource utilization (averaged over different servers) that is achieved on average over an interval of \emph{one hour}, when allocating resources using different variants of $\alpha$PF-VDS.
Again, it can be observed that the proportional fair allocation ($\alpha=1$) results in a greater resource utilization,
while the PS-DSF allocation ($\alpha=\infty$) 
is less efficient compared to other variants.
It is worth noting that the achieved increase in utilization for smaller values of $\alpha$ comes at the price of compromising on fairness.
Intuitively, PS-DSF strives to balance the (weighted) VDS for all users with $x_{n,i}>0$, at each server $i$.
Specifically, the PS-DSF allocation results in $\tilde{s}_{n,i}=\min_m \tilde{s}_{m,i}$, for all users with $x_{n,i}>0,~\forall i$,
provided that $d_{n,r}>0,~\forall n, r$. For an arbitrary allocation ${\bf x}$, we define:
\be
D_n({\bf x}):=\sum_i\frac{x_{n,i}}{x_n} \left(\frac{\tilde{s}_{n,i}-\min_m \tilde{s}_{m,i}}{\min_m\tilde{s}_{m,i}}\right),
\ee
as a measure for \emph{deviation} of each user $n$ from the fair allocation (given by PS-DSF). In Fig.~\ref{fairness} we report the average deviation, $\sum_n\phi_nD_n({\bf x})/\sum_n\phi_n$, and also the maximum deviation among different users, $\max_nD_n({\bf x})$, for our previous experiment in Fig.~\ref{fig:NumExmp:11}. It can be observed that a larger deviation is experienced for smaller values of $\alpha$.
This in turn results in variations in the quality of service (e.g. per-quantum delay) experienced by different jobs.
In Fig.~\ref{Dev_Delay} we report the average per-quantum delay and its standard deviation which are experienced by \emph{short-duration jobs} (those requiring less than one execution quantum) under different variants of the $\alpha$PF-VDS mechanism. Again a larger deviation is experienced for smaller values of $\alpha$.

\begin{figure}[h!]
\centering
\includegraphics[width = 3.3in]{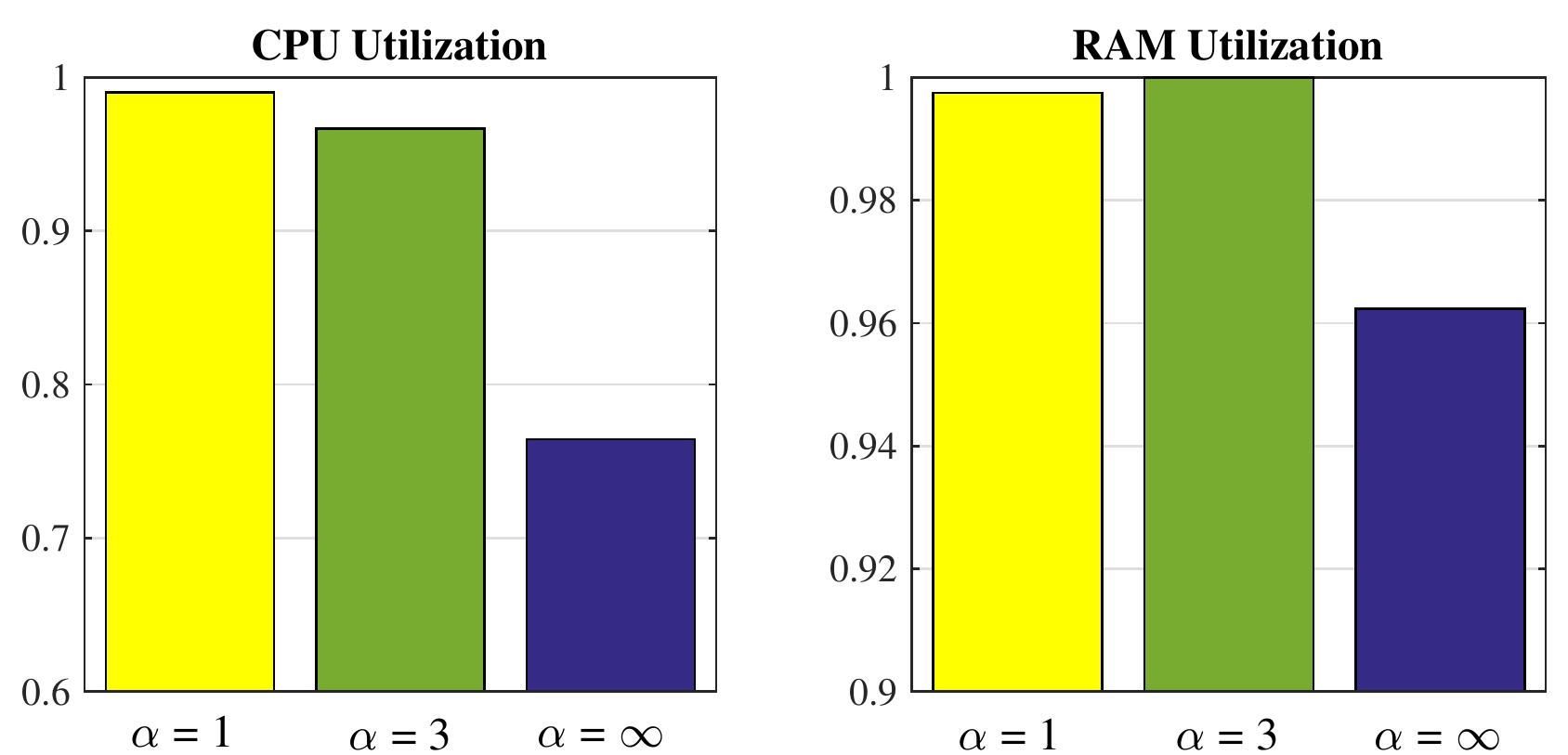}
\footnotesize
\caption{\footnotesize The overall resource utilization that is achieved (on average) over different servers during a 1~hour period,
when allocating resources using $\alpha$PF-VDS with: $\alpha=1$ (proportional fairness), $\alpha=3$, $\alpha=\infty$ (PS-DSF).}
\label{fig:NumExmp:11}
\vspace{-3mm}
\end{figure}

\begin{figure}[h!]
\centering
\includegraphics[width = 3in]{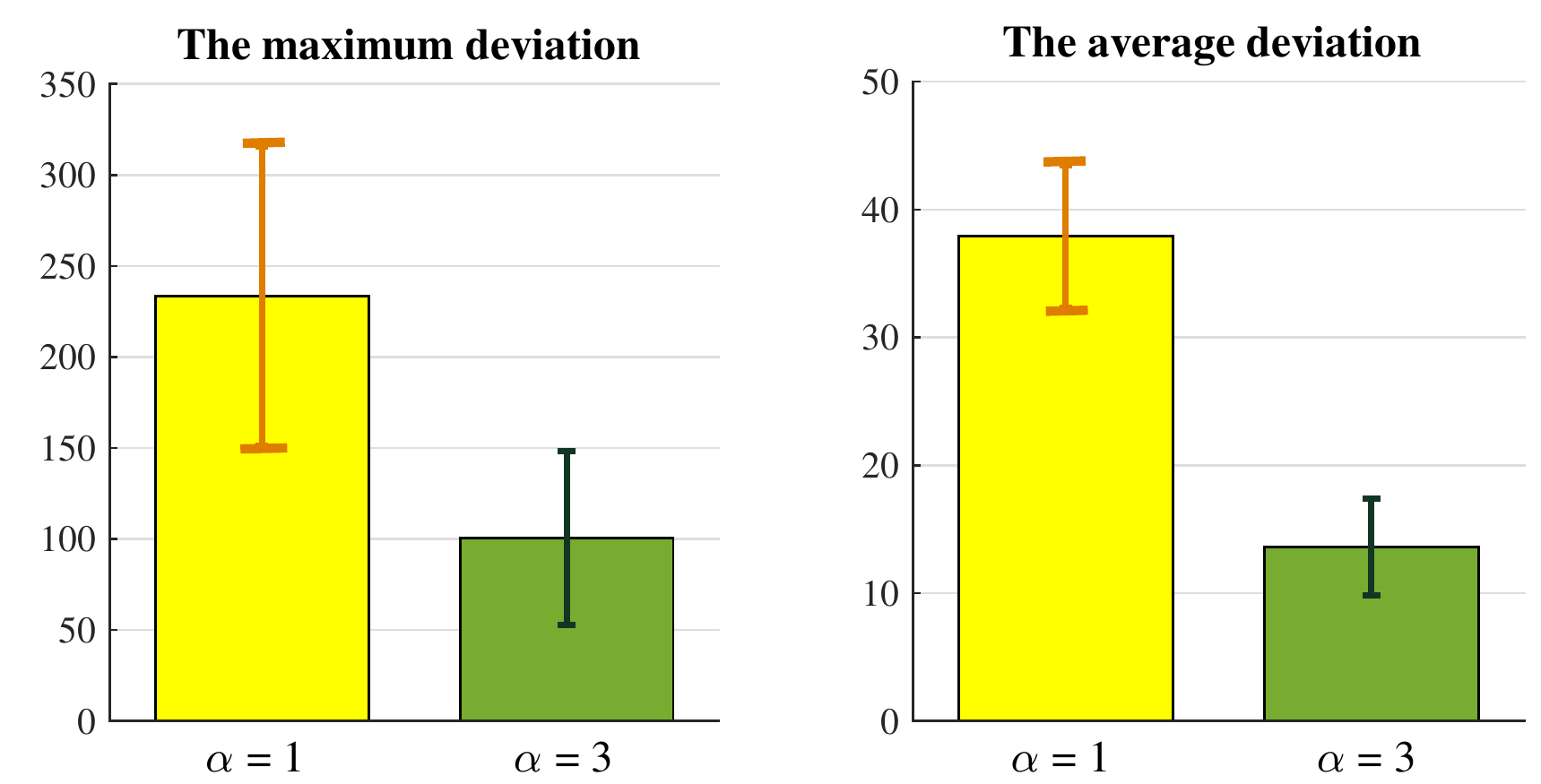}
\footnotesize
\caption{\footnotesize The average and the maximum deviation (averaged over a 1~hour period),
for different variants of the $\alpha$PF-VDS allocation mechanism.}
\label{fairness}
\vspace{-3mm}
\end{figure}

\begin{figure}[h!]
\centering
\includegraphics[width= 2.4in]{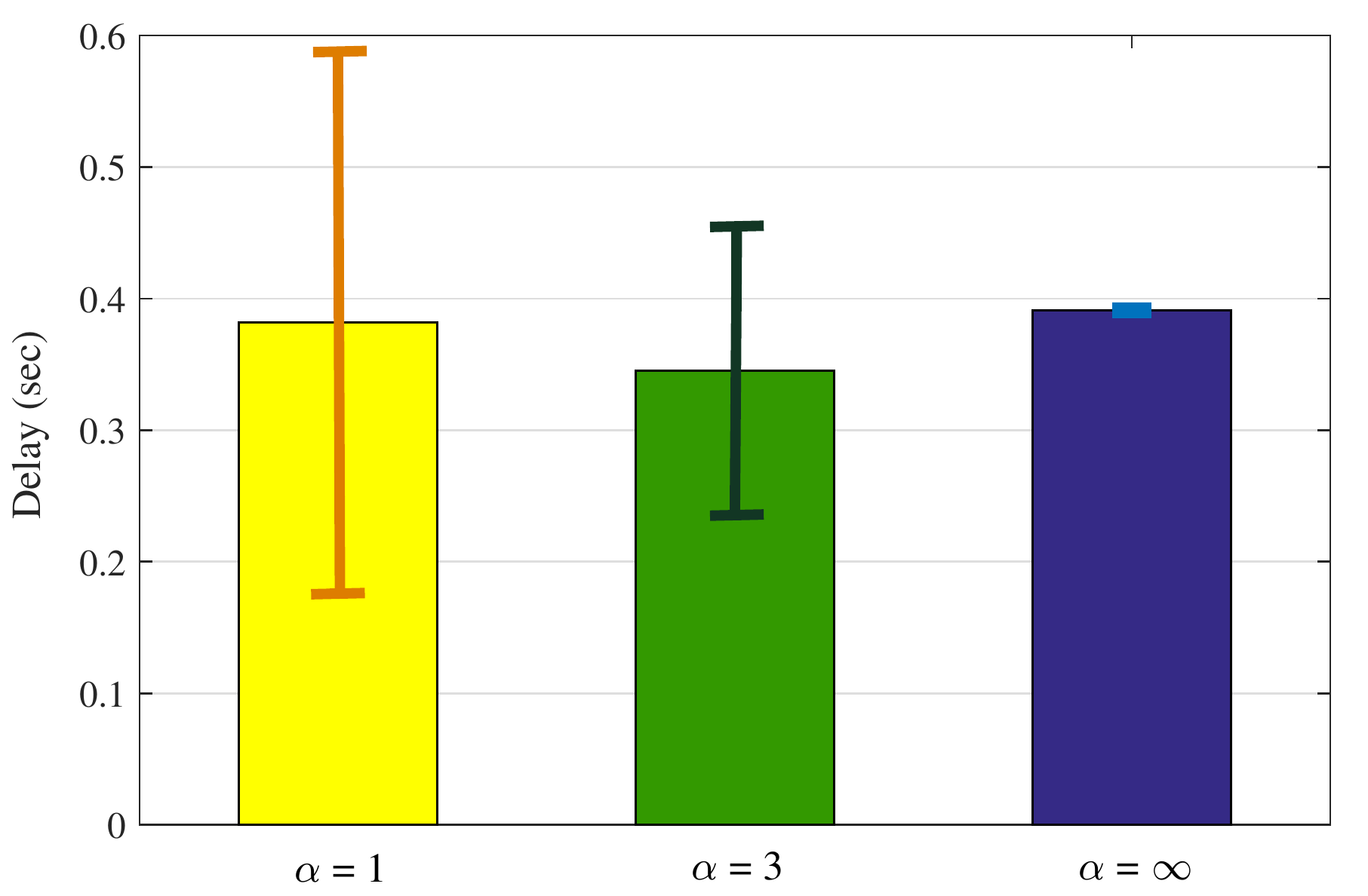}
\footnotesize
\caption{\footnotesize The average and the standard deviation of the per-quantum delay under different variants of
the $\alpha$PF-VDS allocation mechanism.}
\label{Dev_Delay}
\vspace{-2mm}
\end{figure}


\subsection{Comparison with existing mechanisms}
In this subsection, we compare our proposed multi-resource allocation mechanism in terms of resource utilization against the existing multi-resource allocation mechanisms. Specifically, we compare the PS-DSF allocation (which is the least efficient variant of $\alpha$PF-VDS), against the DRFH and TSF allocation mechanisms (which all are applicable to heterogeneous servers in the presence of placement constraints).
First, consider the computing cluster in Fig.~\ref{fig:num_exmp2} feeded by the Bitbrain workload.
We employ each of the aforementioned mechanisms to allocate resources of the servers in Fig.~\ref{fig:num_exmp2} to different sets of active users chosen at random instants of time. The overall utilization that is achieved by of each of these mechanisms for different resources, is shown in Fig.~\ref{fig:diff_alfo}, when averaged over different servers and over 100 runs.
It cane be observed that the PS-DSF allocation mechanism outperforms the two other mechanisms in terms of the achieved utilization for different resources. In particular, the resource utilization is enhanced by the PS-DSF mechanism for up to 20$\%$ for some resources.

\begin{figure}[h!]
\centering
\includegraphics[width = 0.95\columnwidth]{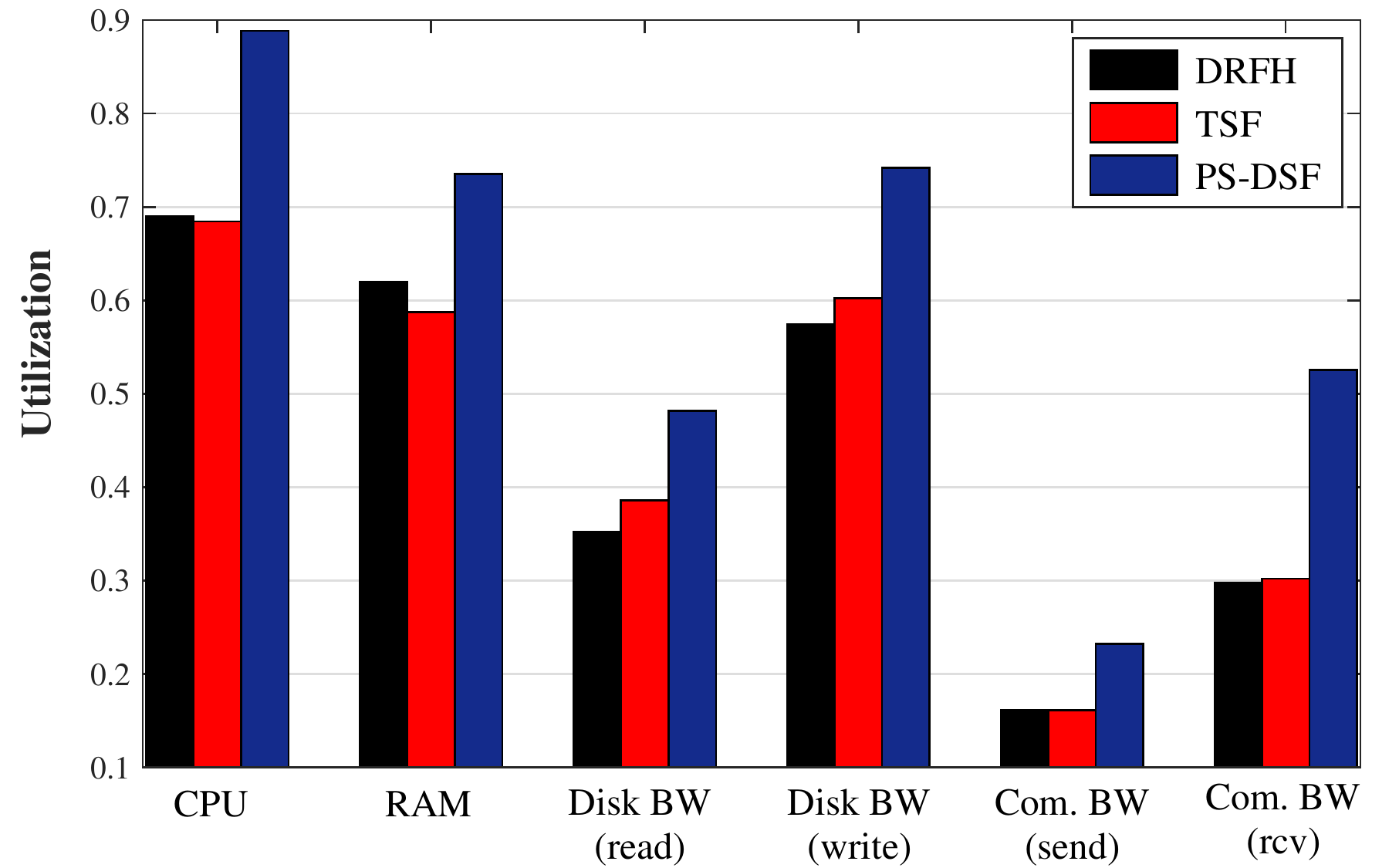}
\footnotesize
\caption{\footnotesize The overall resource utilization, averaged over different servers and over different runs, for different allocation mechanisms.}
\label{fig:diff_alfo}
\vspace{-1mm}
\end{figure}

We make similar observations with the Google traces.
Specifically, consider again the computing cluster in Fig.~\ref{fig:num_exmp},
where 2$\%$ of jobs in the Google traces are randomly chosen as the input workload.
We employ each of the PS-DSF, DRFH and TSF allocation mechanisms to allocate resources of the servers in Fig.~\ref{fig:num_exmp} to demanding jobs over an interval of 24 hours. Fig.~\ref{fig:NumExmp:21} compares the overall resource utilization (averaged over different servers) that is achieved by different allocation mechanisms. It can be observed that PS-DSF is again more efficient in utilizing different resources, compared to the DRFH and TSF allocation mechanisms, while the achieved resource utilization by DRFH and TSF mechanisms is almost the same\footnote{DRFH is just the extension of DRF for multiple heterogeneous servers. There are also other mechanisms, such as those in~\cite{zhu2015EoF}, which approximate DRFH.
As their utilization is the same as, or inferior to the DRFH allocation, we do not consider them herein.}.
The overall resource utilization that is achieved on average over the 24 hour period is shown in Fig.~\ref{fig:NumExmp:23} for different allocation mechanisms. The resource utilization over the last two classes of servers is also shown in Fig.~\ref{fig:NumExmp:23}. It can be observed that the PS-DSF allocation improves the resource utilization over the last two classes of servers more significantly.

\begin{figure}[h!]
\centering
\includegraphics[width=0.99\columnwidth]{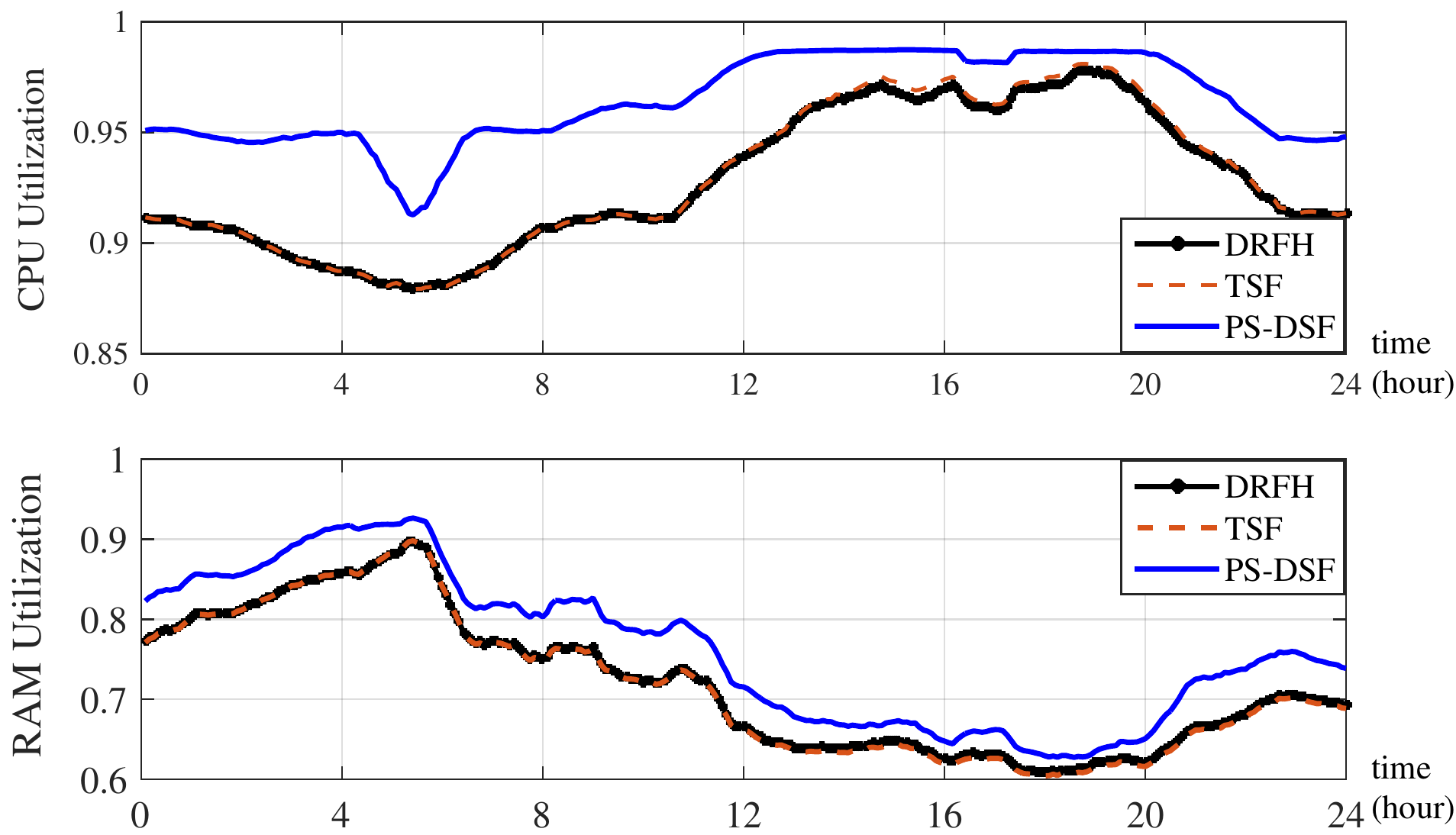}
\footnotesize
\caption{\footnotesize The overall resource utilization (averaged over different servers) that is achieved by different allocation mechanisms during an interval of 24 hours. To get a better view, a moving average with a window size of 1~hour is applied to all graphs.}
\label{fig:NumExmp:21}
\vspace{-2mm}
\end{figure}

\begin{figure}[h!]
\centering
\includegraphics[width=0.98\columnwidth]{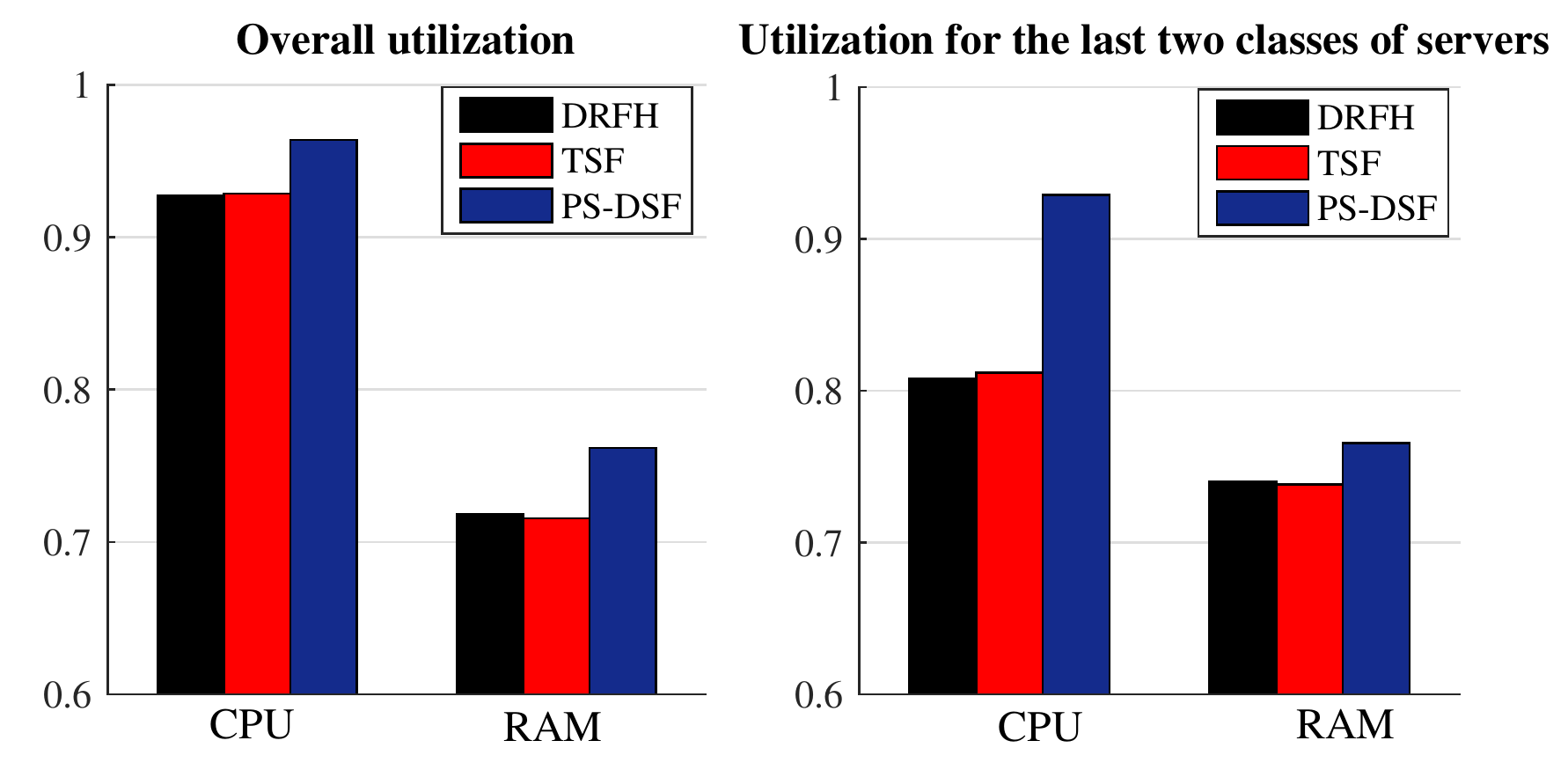}
\footnotesize
\caption{\footnotesize The overall resource utilization that is achieved on average over an interval of 24 hours.}
\label{fig:NumExmp:23}
\vspace{-1mm}
\end{figure}

Intuitively, the PS-DSF allocation mechanism allocates resources at each server based on the \emph{per-server virtual dominant shares}. So, at each server it gives more priority to users which may run more tasks (c.f.~\eqref{VDS}). Hence, one may expect that the PS-DSF allocation mechanism results in a greater resource utilization compared to the DRFH and TSF mechanisms, especially when the resources are heterogeneously distributed over different servers. That is the reason why a more significant increase in utilization is achieved by the PS-DSF allocation over the last two classes of servers, where the resources are more heterogeneously distributed (the available resources over the first two classes of servers in Fig.~\ref{fig:num_exmp} are almost proportional to the overall resource capacities).
This is also consistent with our observation in the first experiment (with the Bitbrain workload),
where the variety of resources along with the heterogeneity of servers results in a significant outperformance by the PS-DSF mechanism.

\section{Conclusion}

We studied efficient and fair allocation of \emph{multiple types} of resources in an environment of \emph{heterogeneous servers} in the presence of \emph{placement constraints}.
We identified potential limitations in the existing multi-resource allocation mechanisms, DRF and its follow up work, when used in such environments. In certain occasions, they may not succeed in satisfying all of the essential fairness-related properties, may not readily be implementable in a distributed fashion, and may lead to inefficient resource utilization. We proposed a \emph{new server-based} approach to efficiently allocate resources while capturing server heterogeneity. We showed how our proposed mechanism could be parameterized (by $\alpha$) to adjust the trade-off between efficiency and fairness. Our resource allocation mechanism was shown to satisfy all of the essential fairness-related properties, i.e., sharing incentive, envy-freeness and bottleneck fairness, for every $\alpha\ge1$, and Pareto optimality for the case $\alpha=1$.
The amenability to distributed implementation usually comes at the price of degrading the performance.
Our proposed mechanism not only is amenable to distributed implementation, but also results in an enhanced resource utilization compared to the existing mechanisms. We carried out extensive simulations, driven by real-world traces, to demonstrate the performance improvements of our resource allocation mechanism.

\section*{Acknowledgment}
\addcontentsline{toc}{section}{Acknowledgment}
The Bitbrain data-set, used for numerical experiments, was graciously provided by Bitbrains IT Services Inc. This data-set is publicly available at http://gwa.ewi.tudelft.nl.

\bibliographystyle{IEEEtran}
\bibliography{cloud,scheduling}

\begin{thebibliography}{10}
\providecommand{\url}[1]{#1}
\csname url@samestyle\endcsname
\providecommand{\newblock}{\relax}
\providecommand{\bibinfo}[2]{#2}
\providecommand{\BIBentrySTDinterwordspacing}{\spaceskip=0pt\relax}
\providecommand{\BIBentryALTinterwordstretchfactor}{4}
\providecommand{\BIBentryALTinterwordspacing}{\spaceskip=\fontdimen2\font plus
\BIBentryALTinterwordstretchfactor\fontdimen3\font minus
  \fontdimen4\font\relax}
\providecommand{\BIBforeignlanguage}[2]{{%
\expandafter\ifx\csname l@#1\endcsname\relax
\typeout{** WARNING: IEEEtran.bst: No hyphenation pattern has been}%
\typeout{** loaded for the language `#1'. Using the pattern for}%
\typeout{** the default language instead.}%
\else
\language=\csname l@#1\endcsname
\fi
#2}}
\providecommand{\BIBdecl}{\relax}
\BIBdecl

\bibitem{DRF}
A.~Ghodsi, M.~Zaharia, B.~Hindman, A.~Konwinski, S.~Shenker, and I.~Stoica,
  ``Dominant resource fairness: Fair allocation of multiple resource types,''
  in \emph{Proc. NSDI}, June 2011.

\bibitem{Chiang12}
C.~Joe-Wong, S.~Sen, T.~Lan, and M.~Chiang, ``{Multi-resource allocation:
  Fairness-efficiency tradeoffs in a unifying framework},'' \emph{IEEE/ACM
  Trans. Networking}, vol.~21, no.~6, Dec. 2013.

\bibitem{Grandl}
R.~Grandl, G.~Ananthanarayanan, S.~Kandula, S.~Rao, and A.~Akella,
  ``Multi-resource packing for cluster schedulers,'' \emph{SIGCOMM Rev.},
  vol.~44, no.~4, pp. 455--466, Aug. 2014.

\bibitem{bonald2014}
T.~Bonald and J.~Roberts, ``Enhanced cluster computing performance through
  proportional fairness,'' \emph{Performance Evaluation}, vol.~79, pp.
  134--145, 2014.

\bibitem{DN}
D.~Bertsekas and R.~Gallager, \emph{Data networks}.\hskip 1em plus 0.5em minus
  0.4em\relax Prentice Hall, 1992.

\bibitem{CDRF}
E.~Friedman, A.~Ghodsi, and C.-A. Psomas, ``Strategyproof allocation of
  discrete jobs on multiple machines,'' in \emph{Proc. ACM Conf. on Economics
  and Computation}, 2014.

\bibitem{DRFH15}
W.~Wang, B.~Liang, and B.~Li, ``Multi-resource fair allocation in heterogeneous
  cloud computing systems,'' \emph{IEEE TPDS}, vol.~26, no.~10, pp. 2822--2835,
  Oct 2015.

\bibitem{HUG}
M.~Chowdhury, Z.~Liu, A.~Ghodsi, and I.~Stoica, ``Hug: Multi-resource fairness
  for correlated and elastic demands,'' in \emph{Proc. NSDI}, Mar 2016.

\bibitem{zhu2015}
Q.~Zhu and J.~C. Oh, ``An approach to dominant resource fairness in distributed
  environment,'' in \emph{Proc. IEA-AIE}.\hskip 1em plus 0.5em minus
  0.4em\relax Springer, May 2015.

\bibitem{UDRF}
Y.~Tahir, S.~Yang, A.~Koliousis, and J.~McCann, ``Udrf: Multi-resource fairness
  for complex jobs with placement constraints,'' in \emph{GLOBECOM}, Dec 2015,
  pp. 1--7.

\bibitem{TSF}
W.~Wang, B.~Li, B.~Liang, and J.~Li, ``Multi-resource fair sharing for
  datacenter jobs with placement constraints,'' \emph{SC 2016}.

\bibitem{PSDSF}
J.~Khamse-Ashari, I.~Lambadaris, G.~Kesidis, B.~Urgaonkar, and Y.~Zhao,
  ``Per-server dominant-share fairness (ps-dsf): A multi-resource fair
  allocation mechanism for heterogeneous servers,'' in \emph{Proc. ICC}, May.,
  2017.

\bibitem{bredin2000}
J.~Bredin, R.~T. Maheswaran, C.~Imer, T.~Ba{\c{s}}ar, D.~Kotz, and D.~Rus, ``A
  game-theoretic formulation of multi-agent resource allocation,'' in
  \emph{Proc. Autonomous Agents}, 2000, pp. 349--356.

\bibitem{jalaparti2010}
V.~Jalaparti and G.~D. Nguyen, ``Cloud resource allocation games,'' Tech. Rep.,
  2010.

\bibitem{wei2010}
G.~Wei, A.~V. Vasilakos, Y.~Zheng, and N.~Xiong, ``A game-theoretic method of
  fair resource allocation for cloud computing services,'' \emph{The journal of
  supercomputing}, vol.~54, no.~2, 2010.

\bibitem{xu2014}
X.~Xu and H.~Yu, ``A game theory approach to fair and efficient resource
  allocation in cloud computing,'' \emph{Mathematical Problems in Engineering},
  vol. 2014, 2014.

\bibitem{zhu2016}
Q.~Zhu and J.~C. Oh, ``Learning fairness under constraints: A decentralized
  resource allocation game,'' in \emph{Proc. IEEE ICMLA}.\hskip 1em plus 0.5em
  minus 0.4em\relax IEEE, 2016, pp. 214--221.

\bibitem{zhu2015EoF}
------, ``Equality or efficiency: A game of distributed multi-type fair
  resource allocation on computational agents,'' in \emph{Proc. IEEE/ACM
  WI-IAT}, vol.~2, 2015, pp. 139--142.

\bibitem{Ghodsi13}
A.~Ghodsi, M.~Zaharia, S.~Shenker, and I.~Stoica, ``Choosy: Max-min fair
  sharing for datacenter jobs with constraints,'' in \emph{Proc. ACM EuroSys},
  2013, pp. 365--378.

\bibitem{midrr}
K.~Yap, T.~Huang, Y.~Yiakoumis, S.~Chinchali, N.~McKeown, and S.~Katti,
  ``Scheduling packets over multiple interfaces while respecting user
  preferences,'' in \emph{Proc. ACM coNEXT}, Dec. 2013.

\bibitem{CM4FQ}
J.~Khamse{-}Ashari, I.~Lambadaris, and Y.~Q. Zhao, ``Constrained multi-user
  multi-server max-min fair queuing,'' \emph{http://arxiv.org/abs/1601.04749},
  Jan. 2016.

\bibitem{Ashari17}
J.~Khamse-Ashari, G.~Kesidis, I.~Lambadaris, B.~Urgaonkar, and Y.~Zhao,
  ``Efficient and fair scheduling of placement constrained threads on
  heterogeneous multi-processors,'' in \emph{Proc. DCPerf}, Atlanta, USA, May
  2017.

\bibitem{Ashari2017j}
------, ``Constrained max-min fair scheduling of variable-length packet-flows
  to multiple servers,'' \emph{Annals of Telecom.}, Aug 2017.

\bibitem{DRF12}
D.~Parkes, A.~Procaccia, and N.~Shah, ``Beyond dominant resource fairness:
  Extensions, limitations, and indivisibilities,'' in \emph{Proc. ACM EC},
  Valencia, Spain, June 2012.

\bibitem{Walrand}
J.~Mo and J.~Walrand, ``Fair end-to-end window-based congestion control,''
  \emph{IEEE/ACM Trans. Networking}, vol.~8, no.~5, Oct 2000.

\bibitem{Rosen}
J.~B. Rosen, ``Existence and uniqueness of equilibrium points for concave
  n-person games,'' \emph{Econometrica: Journal of the Econometric Society},
  pp. 520--534, 1965.

\bibitem{kelly1998rate}
F.~P. Kelly, A.~K. Maulloo, and D.~K. Tan, ``Rate control for communication
  networks: shadow prices, proportional fairness and stability,'' \emph{Journal
  of the Operational Research}, 1998.

\bibitem{GNE_KKT}
A.~Dreves, F.~Facchinei, C.~Kanzow, and S.~Sagratella, ``On the solution of the
  kkt conditions of generalized nash equilibrium problems,'' \emph{SIAM Journal
  on Opt.}, vol.~21, no.~3, 2011.

\bibitem{Fischer98}
A.~Fischer, ``New constrained optimization reformulation of complementarity
  problems,'' \emph{Journal of Optimization Theory and Applications}, vol.~97,
  no.~1, pp. 105--117, 1998.

\bibitem{ResolutionNCP}
C.~Geiger and C.~Kanzow, ``On the resolution of monotone complementarity
  problems,'' \emph{Computational Optimization and Applications}, vol.~5,
  no.~2, pp. 155--173, 1996.

\bibitem{Facchinei98}
\BIBentryALTinterwordspacing
F.~Facchinei, ``Structural and stability properties of p0 nonlinear
  complementarity problems,'' \emph{Math. Oper. Res.}, vol.~23, no.~3, pp.
  735--745, Mar. 1998. [Online]. Available:
  \url{http://dx.doi.org/10.1287/moor.23.3.735}
\BIBentrySTDinterwordspacing

\bibitem{Google11}
C.~Reiss, J.~Wilkes, and J.~L. Hellerstein, ``Google cluster-usage traces,''
  2011, \url{http://code.google.com/p/googleclusterdata/}.

\bibitem{shen2015}
S.~Shen, V.~van Beek, and A.~Iosup, ``Statistical characterization of
  business-critical workloads hosted in cloud datacenters,'' in \emph{Proc.
  15th IEEE/ACM CCGrid}, May 2015, pp. 465--474.

\bibitem{boyd}
S.~Boyd and L.~Vandenberghe, \emph{Convex Optimization}.\hskip 1em plus 0.5em
  minus 0.4em\relax Cambrdige University Press, 2004.

\end{thebibliography}

\appendix

\begin{proof}[\bf Proof of Theorem~\ref{Th_basic}]
With $g'_i(z)=z^{-\alpha}$, the partial derivative of $U_i({\bf x})$ in \eqref{P1_1} with respect to $x_{n,i}$,
\be
\frac{\partial U_i({\bf x})}{\partial x_{n,i}}=\frac{g'_i(\tilde{s}_{n,i})}{\gamma_{n,i}}=\frac{1}{{\gamma_{n,i}}~\tilde{s}_{n,i}^\alpha},
~n\in\mathcal{N}_i,~\forall i,
\ee
and $\partial U_i/\partial x_{n,i}=0$ if $n\notin\mathcal{N}_i$.
Hence,
\be
\nabla_{{\bf x}_i} U_i({\bf x})^T({\bf y}_i-{\bf x}_i)=
\sum_{n\in\mathcal{N}_i}\frac{(y_{n,i}-x_{n,i})/\gamma_{n,i}}{\tilde{s}_{n,i}^\alpha}.\label{grad_local}
\ee
In case that ${\bf x}$ is a solution to Problem~\ref{P1},
the inequality in \eqref{fair_criteria} follows directly from the first order optimality condition. Specifically,
\be
\sum_{n\in\mathcal{N}_i}\frac{(y_{n,i}-x_{n,i})/\gamma_{n,i}}{\tilde{s}_{n,i}^\alpha}=
\nabla_{{\bf x}_i} U_i({\bf x})^T({\bf y}_i-{\bf x}_i) \le0,
\ee
for every feasible ${\bf y}_i=[y_{n,i}]$, and for all servers.
Now assume that \eqref{fair_criteria} is established for a feasible allocation ${\bf x}$.
The equality in \eqref{grad_local} implies that $\nabla_{{\bf x}_i} U_i({\bf x})^T({\bf y}_i-{\bf x}_i)\le0$.
This along with concavity of $U_i({\bf x})$ in terms of ${\bf x}_i$ imply that:
\be\nonumber
U_i({\bf y}_i,{\bf x}_{-i}) &\le& U_i({\bf x}_i,{\bf x}_{-i})+ \nabla_{{\bf x}_i} U_i({\bf x})^T({\bf y}_i-{\bf x}_i)\\\nonumber
                            &\le& U_i({\bf x}_i,{\bf x}_{-i}),
\ee
for every feasible ${\bf y}_i$, and for all servers. This in turn implies that ${\bf x}$ is
a solution (Nash equilibrium) to Problem~\ref{P1}.
\end{proof}

\begin{proof}[\bf Proof of Theorem~\ref{Th_psdsf}]
For $\alpha=1$, we may sum \eqref{fair_criteria} over different servers. Then, for every feasible allocation ${\bf y}$:
\be
\sum_n\phi_n\frac{y_n-x_n}{x_n}\le0,
\ee
which implies that ${\bf x}$ satisfies (weighted) proportional fairness.
In case that $\alpha\rightarrow\infty$, the proof would be similar to that for Lemma~5 in ~\cite{Walrand}.
Let ${\bf x}(\alpha)$ denote an allocation satisfying $\alpha$PF-VDS.
Since the feasible region, described by \eqref{P1_2}-\eqref{P1_4}, is a compact set,
we can find a sequence of $\alpha$, $\{\alpha_l,~l\ge1\}$, such that
$\lim_{l\rightarrow\infty}\alpha_l=\infty$ and $\{{\bf x}(\alpha_l)\}$
converges to some feasible ${\bf x}(\infty)$ as $l\rightarrow\infty$.
By definition, if ${\bf x}(\alpha_l)$ satisfies $\alpha$PF-VDS,
then for every server $i$ and for all feasible allocations, ${\bf y}_{i}$,
\be\nonumber
\sum_{m\in\mathcal{N}_i}\frac{(y_{m,i}-x_{m,i})}{\gamma_{m,i}[\tilde{s}_{m,i}({\alpha_l})]^{\alpha_l}}\le0,
\ee
where $\tilde{s}_{m,i}({\alpha_l}):=x_{m}({\alpha_l})/\phi_m\gamma_{m,i}$ is the weighted VDS for user $m$ with respect to server $i$.
Consider an arbitrary user $n$ and server $i$ for which $y_{n,i}\neq x_{n,i}$. It follows that:
\be\nonumber
\Delta_{n,i}:=\frac{(y_{n,i}-x_{n,i})}{\gamma_{n,i}[\tilde{s}_{n,i}({\alpha_l})]^{\alpha_l}}\le
-\sum_{m\neq n,~m\in\mathcal{N}_i}\frac{(y_{m,i}-x_{m,i})}{\gamma_{m,i}[\tilde{s}_{m,i}({\alpha_l})]^{\alpha_l}},
\ee
Dividing both sides of the above inequality by $\Delta_{n,i}\neq 0$,
\be\nonumber
1 &\le& \sum_{m\neq n,~m\in\mathcal{N}_i} h_{m,i} \left[\frac{\tilde{s}_{n,i}({\alpha_l})}{\tilde{s}_{m,i}({\alpha_l})}\right] ^{\alpha_l}\\
  &\le& \sum_{m\neq n,~h_{m,i}>0} h_{m,i}\left[\frac{\tilde{s}_{n,i}({\alpha_l})}{\tilde{s}_{m,i}({\alpha_l})}\right]^{\alpha_l}\label{Th2_3},
\ee
where $h_{m,i}:=-\frac{\gamma_{n,i}(y_{m,i}-x_{m,i})}{\gamma_{m,i}(y_{n,i}-x_{n,i})}<\infty$.

Unless there exists some user $p$ with $h_{p,i}>0$ and $\tilde{s}_{p,i}(\infty)\le \tilde{s}_{n,i}(\infty)$,
the right hand side of \eqref{Th2_3} approaches zero as $\alpha_l\rightarrow\infty$.
Hence, for the inequality in \eqref{Th2_3} to hold, there must exist some user $p$ with $h_{p,i}>0$ such that $\tilde{s}_{p,i}(\infty)\le \tilde{s}_{n,i}(\infty)$. Note that $h_{p,i}>0$ implies that if $y_{n,i}>x_{n,i}(\infty)$, then $y_{p,i}< x_{p,i}(\infty)$. In other words, for any feasible allocation ${\bf y}\neq{\bf x}(\infty)$,
we may not increase the allocated tasks to user $n$ from server $i$, $y_{n,i}>x_{n,i}(\infty)$,
unless decreasing the allocated tasks from server $i$, $y_{p,i}<x_{p,i}(\infty)$,
for some user $p$ with $\tilde{s}_{p,i}(\infty)\le \tilde{s}_{n,i}(\infty)$.
This implies that ${\bf x}(\infty)$ satisfies PS-DSF.
\end{proof}

\begin{proof}[\bf Proof of Theorem~\ref{Th_Properties}]
For $\alpha=1$, the $\alpha$PF-VDS allocation reduces to a proportional fair allocation,
maximizing a \emph{global objective}, $\sum_n\phi_n\log(x_n)$. So, the resulting allocation is Pareto optimal.  
To prove sharing incentive and envy freeness properties, first we need to derive some inequalities.

As in Section~\ref{sec:reformulate},
let $\lambda_{i,r}$, and $\nu_{n,i}$, denote the Lagrange multipliers corresponding to the constraints in \eqref{P1_2}, and \eqref{P1_3}, respectively. Given the Lagrange multipliers, the KKT conditions for Problem~\ref{P1} are described by \eqref{KKT1_1}-\eqref{KKT1_3}.
Specifically, the first order optimality condition in \eqref{KKT1_3} implies that:
\be
{\gamma_{n,i}^{-1}}{g'_i(\tilde{s}_{n,i})}-\sum_r\lambda_{i,r}d_{n,r} + \nu_{n,i} = 0,~n\in\mathcal{N}_i,~\forall i, \label{Th3_1}
\ee
where $\tilde{s}_{n,i}=x_n/\phi_n\gamma_{n,i}$. 
Since $\nu_{n,i}\ge0$,
\be
{\gamma_{n,i}^{-1}}{g'_i(\tilde{s}_{n,i})} \le \sum_r\lambda_{i,r}d_{n,r},~n\in\mathcal{N}_i,~\forall i. \label{Th3_2}
\ee
By definition, $\gamma_{n,i}=\min_r\{c_{i,r}/d_{n,r}\}$. Hence,
\be
g'_i(\tilde{s}_{n,i}) \le \sum_r\lambda_{i,r}\gamma_{n,i}{d_{n,r}}
                                                        \le \sum_r\lambda_{i,r}c_{i,r},~n\in\mathcal{N}_i,~\forall i. \label{Th3_3}
\ee

Multiplying both sides of \eqref{Th3_1} by $x_{n,i}$, and summing over different users:
\be\nonumber
\sum_{n\in\mathcal{N}_i} \frac{x_{n,i}}{\gamma_{n,i}} g'_i(\tilde{s}_{n,i})
 &=&  \sum_{n\in\mathcal{N}_i}\left[\sum_r\lambda_{i,r}x_{n,i}d_{n,r} + x_{n,i}\nu_{n,i}\right] \\
 &=&  \sum_r\sum_{n\in\mathcal{N}_i}\lambda_{i,r}x_{n,i}d_{n,r},~\forall i, \label{Th3_4}
\ee
where the second equality follows from the fact that $x_{n,i}\nu_{n,i}=0,~n\in\mathcal{N}_i,~\forall i$ (c.f. \eqref{KKT1_2}).
Moreover, the complementary condition in \eqref{KKT1_1} implies that:
\be
\sum_r\sum_{n\in\mathcal{N}_i}\lambda_{i,r}x_{n,i}d_{n,r}=\sum_r\lambda_{i,r}c_{i,r},~\forall i.\label{Th3_5}
\ee
From \eqref{Th3_3}, \eqref{Th3_4} and \eqref{Th3_5} it follows that:
\be
g'_i(\tilde{s}_{n,i}) \le \sum_{m\in\mathcal{N}_i} \frac{x_{m,i}}{\gamma_{m,i}} g'_i(\tilde{s}_{m,i}),
~n\in\mathcal{N}_i,~\forall i.\label{Th3_6}
\ee

Now, we are ready to prove the sharing incentive property. Let $n_i^*:=\argmin_n \tilde{s}_{n,i}$, and
multiply both sides of the inequality in \eqref{Th3_6} for server $i$ and user $n_i^*$ by $(\tilde{s}_{n^*_i,i})^{\alpha-1}$.
For $g'_i(z)=z^{-\alpha}$ and $\alpha\ge1$, it follows that:
\be
(\tilde{s}_{n^*_i,i})^{-1} \nonumber
&\le& \sum_{m\in\mathcal{N}_i} \phi_m\frac{x_{m,i}}{x_{m}}\left[\frac{\tilde{s}_{n^*_i,i}}{\tilde{s}_{m,i}}\right]^{\alpha-1}\\
&\le& \sum_{m\in\mathcal{N}} \phi_m\frac{x_{m,i}}{x_{m}},~\forall i, \label{TH3_7}
\ee
where the second inequality follows from the facts that $\tilde{s}_{n^*_i,i}\le{\tilde{s}_{m,i}},~\forall m$,
and $x_{m,i}=0,~n\notin\mathcal{N}_i$. Furthermore,
\be
(\tilde{s}_{n,i})^{-1} \le (\tilde{s}_{n^*_i,i})^{-1} \le \sum_m \phi_m\frac{x_{m,i}}{x_{m}}, ~\forall n,i.\label{Th3_8}
\ee
Summing \eqref{Th3_8} over different servers,
\be
\frac{\phi_n}{x_n}\sum_i\gamma_{n,i} \le \sum_m\sum_i\phi_m\frac{x_{m,i}}{x_m}=\sum_m\phi_m,\label{Th3_9}
\ee
or,
\be
x_n\ge \frac{\phi_n}{\sum_m\phi_m}\sum_i\gamma_{n,i}.\label{Th3_10}
\ee

To show envy freeness, (by contradiction) assume that user $n$ envies user $m$'s allocation vector,
when adjusted according to their weights. That is, $U_n(\frac{\phi_n}{\phi_m}{\bf a}_m)>U_n({\bf a}_n)$,
where ${\bf a}_m=x_m{\bf d}_m$ and ${\bf a}_n=x_n{\bf d}_n$.
It follows that (c.f.~\eqref{utility}):
\be
\frac{x_nd_{n,r}}{\phi_n} < \frac{x_md_{m,r}}{\phi_m} ,~\forall r.\label{Th3_11}
\ee
Consider some server $j$ for which 
$x_{m,j}>0$, so that $\nu_{m,j}=0$.
Then, \eqref{Th3_1} implies that:
\be
\sum_r\lambda_{j,r}d_{m,r} = {\gamma_{m,j}^{-1}}{g'_j(\tilde{s}_{m,j})}. \label{Th3_1'}
\ee
If we multiply both sides of \eqref{Th3_2} by $x_n/\phi_n$,
then it follows from \eqref{Th3_11} and \eqref{Th3_1'} that:
\be\nonumber
\tilde{s}_{n,j}g'_j(\tilde{s}_{n,j}) &\le& \sum_r\frac{\lambda_{j,r}d_{n,r}x_n}{\phi_n}\\\nonumber
                             & < & \sum_r\frac{\lambda_{j,r}d_{m,r}x_m}{\phi_m}\\
                             & = & \tilde{s}_{m,j}g'_j(\tilde{s}_{m,j})\label{Th3_12}
\ee
In case that $g'_j(z)=z^{-\alpha}$ and $\alpha>1$, it follows from \eqref{Th3_12} that
$\tilde{s}_{n,j}>\tilde{s}_{m,j}$, and therefore,
\be
\frac{x_nd_{n,\rho(n,j)}}{\phi_nc_{j,\rho(n,j)}}=\frac{x_n}{\phi_n\gamma_{n,j}} > \frac{x_m}{\phi_m\gamma_{m,j}}
\ge\frac{x_md_{m,r}}{\phi_mc_{j,r}},~\forall r,\label{Th3_13}
\ee
which contradicts \eqref{Th3_11} for $r=\rho(n,j)$.
In case that $g'_j(z)=z^{-1}$, \eqref{Th3_12} requires that $1<1$, which is a contradiction.
\end{proof}

\begin{proof}[\bf Proof of Theorem~\ref{Th_Universal_Allocation}]
A resource, say $\rho(i)$, is identified as a bottleneck resource at server $i$, if it is dominantly requested by all eligible users for server $i$. That is $\rho(n,i)=\rho(i),~n\in\mathcal{N}_i$.
Hence,
\be
\sum_{n\in\mathcal{N}_i} x_{n,i}\frac{d_{n,r}}{c_{i,r}}\le \sum_{n\in\mathcal{N}_i}x_{n,i}\frac{d_{n,\rho(i)}}{c_{i,\rho(i)}} =
\sum_{n\in\mathcal{N}_i}\frac{x_{n,i}}{\gamma_{n,i}},~\forall i,r.\label{resource_usage}
\ee
The above inequality implies that all of the capacity constraints for server $i$ could be replaced by:
\be
\sum_{n\in\mathcal{N}_i}\frac{x_{n,i}}{\gamma_{n,i}}\le1. \label{new_cap_const}
\ee
Let $\lambda_i$ denote the Lagrange multiplier corresponding to the constraint in \eqref{new_cap_const}.
Hence, the KKT conditions for Problem~\ref{P1} (given by \eqref{KKT1_1}-\eqref{KKT1_3}) can be simplified/rewritten as:
\be
&& g'_i(\frac{x_n}{\phi_n\gamma_{n,i}}) - \lambda_{i} + \nu'_{n,i} = 0,~n\in\mathcal{N}_i,~\forall i,\qquad\label{KKT3_1}\\
&& 0\le\lambda_{i} \perp (1-\sum_{n\in\mathcal{N}_i}{x_{n,i}}/{\gamma_{n,i}})\ge0,~\forall i,\label{KKT3_2}\\
&& 0\le\nu'_{n,i} \perp x_{n,i}\ge0,~n\in\mathcal{N}_i,~\forall i,\label{KKT3_3}
\ee
where, $\nu'_{n,i}:=\nu_{n,i}\gamma_{n,i}$.

First assume that ${\bf x}$ is a solution (Nash equilibrium) for Problem~\ref{P1}, so that the KKT conditions in \eqref{KKT3_1}-\eqref{KKT3_3} are satisfied. We know that $g'(\cdot)>0$, owing to the assumption that $g(\cdot)$ is strictly increasing and differentiable.
This along with \eqref{KKT3_1} imply that $\lambda_{i}>\nu'_{n,i},~n\in\mathcal{N}_i~\forall i$,
and therefore, $\lambda_{i}>0, \forall i$.
According to \eqref{KKT3_2},
\be
\sum_{n\in\mathcal{N}_i} \frac{x_{n,i}}{\gamma_{n,i}} = 1,~\forall i. \label{tight_bottleneck}
\ee
The assumption that $g(\cdot)$ is \emph{strictly concave} implies that $g'(\cdot)$ is strictly decreasing, and therefore is invertible.
The inverse of $g'(\cdot)$ would be also a strictly decreasing function which we denote it by $h(\cdot)$.
It follows from \eqref{KKT3_1} that:
\be
\tilde{s}_{n,i} := \frac{x_n}{\phi_n\gamma_{n,i}} = h({\lambda_{i,}-\nu'_{n,i}}),~n\in\mathcal{N}_i~\forall i.
\ee

Consider two users $n,m\in\mathcal{N}_i$, where $x_{n,i}>0$, so that $\nu'_{n,i}=0$ (c.f. \eqref{KKT3_3}). It follows that,
\be\label{conseq}
\tilde{s}_{n,i}=h({\lambda_{i}})\le h({\lambda_{i}-\nu'_{m,i}})=\tilde{s}_{m,i},
\ee
where the above inequality follows from the fact that $h(\cdot)$ is a decreasing function.
According to \eqref{tight_bottleneck} and \eqref{conseq}, we may not increase $x_{m,i}$ for any user $m$,
unless decreasing $x_{n,i}$ for some user $n$ with $\tilde{s}_{n,i}\le\tilde{s}_{m,i}$.
So, by definition, ${\bf x}$ satisfies PS-DSF.

Now consider an allocation ${\bf x}$ satisfying PS-DSF.
For such an allocation, there exists (at least) one resource at each server $i$ for which the capacity constraint holds with equality.
This along with \eqref{resource_usage} imply that the constraint in \eqref{new_cap_const} holds with equality at each server $i$.
Hence, \eqref{KKT3_2} is established for all servers.
For each server $i$, consider some user $n$ for which $x_{n,i}>0$.
By definition, $\tilde{s}_{p,i}\ge\tilde{s}_{n,i},~\forall p$, and $\tilde{s}_{p,i}=\tilde{s}_{n,i},~\forall p:x_{p,i}>0$.
Hence, we may choose $\lambda_{i}:=g'(\tilde{s}_{n,i})$, and
\be
\nu'_{m,i}=\lambda_{i}-g'_i(\tilde{s}_{m,i})\ge0,~m\in\mathcal{N}_i,~\forall i,\label{slack_var}
\ee
so that the first order optimality condition in \eqref{KKT3_1} is established.
Furthermore, \eqref{slack_var} implies that $\nu'_{m,i}=0$ when $x_{m,i}>0$.
So, all of the conditions in \eqref{KKT3_1}-\eqref{KKT3_3} are established in conjunction with the chosen multipliers.
In other words, given a PS-DSF allocation, ${\bf x}$, we may find a set of multipliers such that the KKT conditions in \eqref{KKT3_1}-\eqref{KKT3_3} are established. This implies that ${\bf x}$ is a solution to Problem~\ref{P1}.
\end{proof}

\begin{proof}[\bf Proof of Theorem~\ref{Th_non_divisible_servers}]
The proof follows exactly the same line of arguments as Theorem~\ref{Th_Universal_Allocation}.
Specifically, let $\lambda_{i}$, and $\nu_{n,i}$, denote the Lagrange multipliers corresponding to the constraints in \eqref{P2_2}, and \eqref{P2_3}. The Lagrangian function for Problem~\ref{P2} at server $i$ is given by:
\be\nonumber
\mathcal{L}_i({\bf x},{\bf \lambda},{\bf\nu}) := \sum_{n\in\mathcal{N}_i}\left[\phi_ng_i(\frac{x_n}{\phi_n\gamma_{n,i}})+\sum_n\nu_{n,i}x_{n,i}\right]\\
  + ~\lambda_{i}\left[1-\sum_{n\in\mathcal{N}_i}\frac{x_{n,i}}{\gamma_{n,i}}\right].\qquad\qquad\qquad
\ee
It can be observed that the KKT conditions for this problem are exactly the same as those given by \eqref{KKT3_1}-\eqref{KKT3_3} in Theorem~\ref{Th_Universal_Allocation}.
Hence, the same arguments apply here.
\end{proof}

\begin{proof}[\bf Proof of Corollary~\ref{Corollary_EF_SI}]
As in Theorem~\ref{Th_Properties}, let $n_i^*:=\argmin_n \tilde{s}_{n,i}$. Then, divide both sides of the inequality in \eqref{Th3_6}
for server $i$ and user $n_i^*$ by $g'_i(\tilde{s}_{n^*_i,i})\tilde{s}_{n^*_i,i}$. It follows that:
\be
(\tilde{s}_{n^*_i,i})^{-1}
&\le& \sum_{m\in\mathcal{N}_i} \phi_m\frac{x_{m,i}}{x_{m}}\left[\frac{\tilde{s}_{m,i}}{\tilde{s}_{n^*_i,i}}
                                              \frac{g'_i(\tilde{s}_{m,i})}{g'_i(\tilde{s}_{n^*_i,i})}\right].\label{Coro_EFSI_1}
\ee
To simplify the right hand side of \eqref{Coro_EFSI_1}, let define:
\be
q_i(z):=\frac{1}{zg'_i(z)}=\frac{(\frac{z}{A_i})^{\alpha_i-1}}{A_i +{B_i}(\frac{z}{{A_i}})^{\alpha_i-1}}. \label{Coro_EFSI_2}
\ee
It can be observed that $q_i(z)$ is an increasing function. Hence,
\be\nonumber
(\tilde{s}_{n,i})^{-1}\le(\tilde{s}_{n^*_i,i})^{-1}
&\le& \sum_{m\in\mathcal{N}_i} \phi_m\frac{x_{m,i}}{x_{m}}\left[\frac{q_i(\tilde{s}_{n^*_i,i})}{q_i(\tilde{s}_{m,i})}\right]\\
&\le& \sum_{m\in\mathcal{N}_i} \phi_m\frac{x_{m,i}}{x_{m}},~\forall n, i, \label{Coro_EFSI_3}
\ee
where the first and the last inequalities follow from the fact that $\tilde{s}_{n^*_i,i}\le{\tilde{s}_{m,i}},~\forall m$.
As in Theorem~\ref{Th_Properties}, we may sum \eqref{Coro_EFSI_3} over different servers to get the required result (c.f. \eqref{Th3_9} and \eqref{Th3_10}).
\end{proof}

\begin{proof}[\bf Proof of Corollary~\ref{Corollary_PSDSF}]
Consider a sequence, $\{\alpha_l,~l\ge1\}$, which converges to infinity, $\lim_{l\rightarrow\infty}\alpha_l=\infty$,
and $\{{\bf x}(\alpha_l)\}$ converges to some feasible ${\bf x}(\infty)$ as $l\rightarrow\infty$.
 When ${\bf x}(\alpha_l)$ is a solution to Problem~\ref{P1}, then for every feasible allocation, ${\bf y}$, 
 it can be shown that (c.f. Theorem~\ref{Th_basic}):
\be\nonumber
\sum_m g'_i(\tilde{s}_{m,i}({\alpha_l})) \frac{y_{m,i}-x_{m,i}}{\gamma_{m,i}}\le0,~\forall i.
\ee
Now, consider an arbitrary user $n$ and server $i$ for which $y_{n,i}\neq x_{n,i}$.
As in Theorem~\ref{Th_psdsf}, it can be shown that:
\be
1 &\le& \sum_{m\neq n,~h_{m,i}>0} h_{m,i} \frac{g'_i(\tilde{s}_{m,i}({\alpha_l}))}{g'_i(\tilde{s}_{n,i}({\alpha_l}))}\label{Coro_PSDSF_2},
\ee
where $h_{m,i}:=-\frac{\gamma_{n,i}(y_{m,i}-x_{m,i})}{\gamma_{m,i}(y_{n,i}-x_{n,i})}<\infty$. On the other hand,
\be
\lim_{l\rightarrow\infty}\frac{g'_i(\tilde{s}_{m,i}({\alpha_l}))}{g'_i(\tilde{s}_{n,i}({\alpha_l}))}=
\lim_{l\rightarrow\infty}\left[\frac{\tilde{s}_{n,i}({\alpha_l})}{\tilde{s}_{m,i}({\alpha_l})}\right]^{\alpha_l}, \label{Coro_PSDSF_3}
\ee
provided that $\tilde{s}_{n,i}/A_i\le 1,~n\in\mathcal{N}_i$. The limit in \eqref{Coro_PSDSF_3},
and therefore the right hand side of \eqref{Coro_PSDSF_2}, approaches zero as $\alpha_l\rightarrow\infty$,
unless there exists some user $p$ with $h_{p,i}>0$ and $\tilde{s}_{p,i}(\infty)\le \tilde{s}_{n,i}(\infty)$.
In other words, for the inequality in \eqref{Coro_PSDSF_2} to be established,
there must exist some user $p$ with $h_{p,i}>0$ such that $\tilde{s}_{p,i}(\infty)\le \tilde{s}_{n,i}(\infty)$.
As in Theorem~\ref{Th_psdsf}, this implies that ${\bf x}(\infty)$ satisfies PS-DSF.
\end{proof}

\begin{proof}[\bf Proof of Theorem~\ref{Th_stationary_points}]
Let $\nabla_{\bf x}\Psi=[(\nabla_{\bf x}\Psi)_{n,i}]_{NK\times1}$ denote the gradient of $\Psi$
with respect to ${\bf x}$, and $\nabla_{\bf \lambda}\Psi=[(\nabla_{\bf \lambda}\Psi)_{i,r}]_{MK\times1}$ denote the gradient of $\Psi$ with respect to ${\bf \lambda}$. In particular,
\be
 (\nabla_{\bf x}\Psi)_{n,i} := \frac{\partial\Psi}{\partial x_{n,i}} = 
\left[\frac{\partial\psi_{n,i}}{\partial x_{n,i}} +
\sum_j\frac{\partial\psi_{n,j}}{\partial f_{n,j}}\frac{\partial f_{n,j}}{\partial x_n}\right],\label{partial_1}\\
(\nabla_{\bf \lambda}\Psi)_{i,r} := \frac{\partial\Psi}{\partial \lambda_{i,r}} = \sum_m \frac{\partial\psi_{m,i}}{\partial f_{m,i}}d_{m,r},\qquad\qquad\qquad\label{partial_2}
\ee
where $\psi_{n,i}:=\psi (x_{n,i},f_{n,i}({\bf x},{\bf \lambda}_i)),~n\in\mathcal{N}_i$, and $\psi_{n,i}=0,~n\notin\mathcal{N}_i$. Let
\be
{\bf d}_b\psi:=\left[\frac{\partial \psi_{n,i}}{\partial f_{n,i}}\frac{\partial f_{n,i}}{\partial x_n}\right]_{NK\times 1}.
\ee
$({\bf x},{\bf \lambda})$ is a stationary point of $\Psi$ when
$\nabla\Psi=[\nabla_{\bf x}\Psi;\nabla_{\bf \lambda}\Psi]=0$.
If we pre-multiply $\nabla_{\bf x}\Psi$ by ${\bf d}_b\psi$, it follows that (c.f. \eqref{partial_1}):
\be\nonumber
({\bf d}_b\psi)^T \nabla_{\bf x}\Psi
&=& \sum_{n,i}\frac{\partial\psi_{n,i}}{\partial x_{n,i}}\frac{\partial\psi_{n,i}}{\partial f_{n,i}}\frac{\partial f_{n,i}}{\partial x_n}\\
&+&\sum_n\left(\sum_i\frac{\partial\psi_{n,i}}{\partial f_{n,i}}\frac{\partial f_{n,i}}{\partial x_n}\right)^2=0,\label{descent_direction_test}
\ee
where,
\be
\frac{\partial f_{n,i}}{\partial x_n} :=  \frac{-1}{\phi_n\gamma^2_{n,i}}g''_i(\frac{x_n}{\phi_n\gamma_{n,i}}) > 0,
\ee
owing to strict concavity of $g_i(\cdot)$. For the complementary function $\psi(a,b)$ defined in \eqref{NCP_function}, it is shown that~\cite{ResolutionNCP}:
\be
&&\frac{\partial\psi}{\partial a}\frac{\partial\psi}{\partial b}\ge 0,~\forall (a,b)\in\mathbb{R}^2,\\
&&\frac{\partial\psi}{\partial a}=\frac{\partial\psi}{\partial b}=0~\Longleftrightarrow~\psi(a,b)=0.
\ee
Hence, the right hand side of \eqref{descent_direction_test} is strictly positive, unless ${\partial\psi_{n,i}}/{\partial x_{n,i}}={\partial\psi_{n,i}}/{\partial f_{n,i}}=0,~\forall n,i$, or equivalently
$\psi (x_{n,i},f_{n,i}({\bf x},{\bf \lambda}_i))=0,~\forall n,i$,
which is the case only if ${\bf x}$ is a solution to Problem~\ref{P1}.
Conversely, for every solution to Problem~\ref{P1},
$\psi (x_{n,i},f_{n,i}({\bf x},{\bf \lambda}_i))=0,~\forall n,i$,
which implies that ${\partial\psi_{n,i}}/{\partial x_{n,i}}={\partial\psi_{n,i}}/{\partial f_{n,i}}=0,~\forall n,i$,
and therefore $\nabla\Psi = 0$.
\end{proof}

\begin{proof}[\bf Proof of Theorem~\ref{Th_uniqueness}]
By contradiction, assume that there exists another solution $({\bf x}',{\lambda}')$,
arbitrarily close to $({\bf x},{\lambda})$, for which $x'_m\neq x_m$, for some user $m$.
Since $({\bf x}',{\lambda}')$ is arbitrarily close to $({\bf x},{\lambda})$, by continuity
we conclude that $f_{n,i}({\bf x}',\lambda'_i)>0$ if $f_{n,i}({\bf x},\lambda_i)>0$.
Hence, $x'_{n,i}>0$ only if $f_{n,i}({\bf x},\lambda_i)=0$ (c.f. \eqref{KKT2_2}).
In the same way, $\lambda'_{i,r}>0$ only if $f_{i,r}({\bf x}_i)=0$.
On the other hand, the assumption that $({\bf x}',{\lambda}')$ is a solution to Problem~\ref{P3},
implies that $f_{i,r}({\bf x}'_i)=0$, for every server $i$ and resource $r$ with $\lambda'_{i,r}>0$ (c.f. \eqref{KKT2_1}). Hence,
\be
\Delta f_{i,r} := f_{i,r}({\bf x}'_i) - f_{i,r}({\bf x}_i)=0,~~\forall i,r:\lambda'_{i,r}>0. \label{Th_U_1}
\ee
Furthermore, \eqref{KKT2_2} implies that $f_{n,i}({\bf x}',\lambda'_i)=0$, for every user $n$ and server $i$ with $x'_{n,i}>0$. Hence,
\be
\Delta f_{n,i} := f_{n,i}({\bf x}',\lambda'_i) - f_{n,i}({\bf x},\lambda_i) = 0,~~\forall n,i:x'_{n,i}>0. \label{Th_U_2}
\ee

Let $\delta x_{n,i}:=x'_{n,i}-x_{n,i}$ and $\delta\lambda_{i,r}:=\lambda'_{i,r}-\lambda_{i,r}$.
Since $({\bf x}',{\lambda}')$ could be arbitrarily close to $({\bf x},{\lambda})$,
it follows from \eqref{Th_U_1} and \eqref{Th_U_2} that:
\be
\Delta f_{i,r} &    =   & -\sum_n d_{n,r}\delta x_{n,i} = 0,~~\forall i,r:\lambda'_{i,r}>0,\label{Th_U_3}\\\nonumber
\Delta f_{n,i} & \simeq &
       \frac{\alpha_i \delta x_n}{x_n}\left[\frac{\phi_n^{\alpha_i}\gamma_{n,i}^{\alpha_i-1}A_i^{\alpha_i}}{x_n^{\alpha_i}}\right]+\frac{B_i \delta x_n}{x_n^2}\\
               & + &  \sum_r d_{n,r}\delta\lambda_{i,r}=0,~~\forall n,i:x'_{n,i}>0\label{Th_U_4},
\ee
where $\delta x_n:=\sum_i\delta x_{n,i}$. Let $r^*(i)$ denote the only resource at server $i$ for which $\lambda_{i,r^*(i)}>0$.
The fact that $f_{n,i}({\bf x},\lambda_i)=0$ for every user $n$ with $x'_{n,i}>0$, implies that:
\be
d_{n,r^*(i)}\lambda_{i,r^*(i)} = \frac{\phi_n^{\alpha_i}\gamma_{n,i}^{\alpha_i-1}A_i^{\alpha_i}}{x_n^{\alpha_i}} + \frac{B_i}{x_n}.\label{Th_U_5}
\ee
Hence, we may rewrite \eqref{Th_U_4} as:
\be
 d_{n,r^*(i)}\delta\lambda_{i,r^*(i)} + \omega_{n,i}\delta x_{n} = 0, ~~\forall n,i:x'_{n,i}>0,\label{Th_U_6}
\ee
where, $\omega_{n,i}$ is defined as:
\be
\omega_{n,i} := \left[\frac{\alpha_id_{n,r^*(i)}\lambda_{i,r^*(i)}}{x_n}+\frac{B_i(1-\alpha_i)}{x_n^2}\right].\label{Th_U_7}
\ee

In the following, we check the possibility for existence of $({\bf x}',{\lambda}')$ such that \eqref{Th_U_3} and \eqref{Th_U_6} are established.
In general, one may partition the set of users, and the set of servers respectively, into $L$ partitions,
$\mathcal{N}=\{\mathcal{N}_1, \mathcal{N}_2,\cdots, \mathcal{N}_L\}$ and
$\mathcal{K}=\{\mathcal{K}_1,\mathcal{K}_2,\cdots, \mathcal{K}_L\}$,
such that $x'_{n,i}=0$ for each user $n\in\mathcal{N}_l$ and for every server $i\notin\mathcal{K}_l$.
Here, without loss of generality, we assume that $L=1$ is the greatest number of partitions which could be found [Otherwise, we should consider each partition separately]. 

By assumption, $\sum_i\delta x_{m,i}=\delta x_m\neq 0$ for user $m$. Without loss of generality assume that $\delta x_m>0$.
From \eqref{Th_U_6} it follows that $\delta\lambda_{i,r^*(i)}<0$, for every server $i$ for which $x'_{m,i}>0$.
Furthermore, \eqref{Th_U_6} implies that $\delta x_{n}>0$ for every user $n$,
for which $x'_{m,j}x'_{n,j}>0$ for some server $j$.
Given that all users and all servers reside in the same partition, it follows that
$\delta x_n>0$ for all users, and $\delta\lambda_{i,r^*(i)}<0$ for all servers.

Let $\mathcal{S}$, $S:=|\mathcal{S}|$, denote the set of servers for which $\delta{\bf x}_i=[\delta x_{n,i}]\neq0$.
For every server $i\in\mathcal{S}$, \eqref{Th_U_3} implies that there exists some user $p$ for which $\delta x_{p,i}<0$.
On the other hand, for each user $n$ there exists some server $i$ for which $\delta x_{n,i}>0$.
Hence, $x'_{n,i}>0$ for at least $N+S$ pairs of users and servers, for which \eqref{Th_U_6} should be established.

To present the system of equations in \eqref{Th_U_6} in a matrix form, let define $\chi_i:=\{n\mid x'_{n,i}>0\}$.
Also, define $W_i=\text{diag}([\omega_{n,i}])$, and $V_i=W_i(\{n\in\chi_i\},:)$ as a sub-matrix of $W_i$,
consisting of a subset of rows in $W_i$ which corresponds to users $n\in\chi_i$.
Then, the system of equations in \eqref{Th_U_6} can be written as:
\be
\begin{bmatrix}[cccc|c]
{\bf d}^*_1 &         0     & \cdots &         0     &  V_1  \\
0           &  {\bf d}^*_2  & \cdots &         0     &  V_2  \\
\vdots      &     \vdots    & \ddots &     \vdots    & \vdots\\
0           &         0     & \cdots &  {\bf d}^*_S  &  V_S  \\
\end{bmatrix}
\begin{bmatrix}
\delta{\lambda}_{1,r^*(1)}       \\
\vdots                      \\
\delta{\lambda}_{S,r^*(S)}       \\
\delta{\bf x}
\end{bmatrix} = 0,\label{matrix_system1}
\ee
where ${\bf d}^*_i:=[d_{n,r^*(i)} \mid n\in\chi_i]$, and $\delta{\bf x}:=[\delta x_n\mid n\in\mathcal{N}]$.
Here, without loss of generality, we have assumed that servers in $\mathcal{S}$ are indexed from 1 to $S$.
As a short hand notation, we denote the coefficient matrix in \eqref{matrix_system1} by $[D\mid V]$.
In the following we show that the coefficient matrix in \eqref{matrix_system1} has a column rank of $N+S$.

For the matrix $V$, we know that exactly one element is non-zero in each row.
Furthermore, there exists at least a non-zero element in each column, owing to the fact that for each user $n$,
$x'_{n,i}>0$ for at least one server $i$. Hence, we may find $N$ linear independent rows in $V$, or equivalently $N$ linear independent columns, which form this matrix. It means that $V$ is a full rank matrix which has a column rank of $N$. Furthermore, from \eqref{Th_U_6} and \eqref{Th_U_7}, it can be observed that none of the columns (or group of columns) in $V$ can be expressed as a linear combination of columns of $D$, provided that $B_i>0,~\forall i$. Also, none of the columns (or group of columns) in $D$ can be expressed in terms of columns in $V$, owing to the assumption that all users and all servers reside in the same partition. Accordingly, the coefficient matrix, $[D,V]$, has a column rank of $N+S$. Hence, the only solution to \eqref{matrix_system1} is $\delta{\bf \lambda}_i=0,~\forall i$, and $\delta{\bf x}=0$. However, this contradicts the assumption that $\delta x_{m}\neq0$ for user $m$.
\end{proof}

\begin{proof}[\bf Proof of Lemma~\ref{lem_descent}]
$({\bf v}_{x}^h)_i$ is the projection of $-\nabla_{{\bf x}_i}\Psi$ onto $\Omega_i$,
when:
\be
({\bf v}_{x}^h)_i =\argmin_{{\bf v}_i\in\Omega_i}\frac{1}{2}\|{\bf v}_i+\nabla_{{\bf x}_i}\Psi\|^2.\label{Th8_0}
\ee
According to the constrained optimality theorem, $({\bf v}_{x}^h)_i$ is a solution to minimization in \eqref{Th8_0} if and only if~\cite{boyd}:
\be
[({\bf v}_{x}^h)_i+\nabla_{{\bf x}_i}\Psi]^T[{\bf v}_i-({\bf v}_{x}^h)_i]\ge0,\label{Th_8_1}
\ee
for every ${\bf v}_i\in\Omega_i$. Substituting ${\bf v}_i=0$ results in:
\be
({\bf v}_{x}^h)_i^T\nabla_{{\bf x}_i}\Psi\le -\|({\bf v}_{x}^h)_i\|^2
\ee
which implies that ${\bf v}_{x}^h$ is a strictly descent direction, unless ${\bf v}_{x}^h = 0$.
On the other hand, for every $i,r$ with $({\bf v}_{\lambda}^h)_{i,r}\neq0$, it follows from \eqref{lambda_non_feas_dir} and \eqref{lambda_feas_dir} that:
\be\nonumber
 \frac{\partial\Psi}{\partial\lambda_{i,r}} \times ({\bf v}_{\lambda}^h)_{i,r} = -\left(\frac{\partial\Psi}{\partial\lambda_{i,r}}\right)^2\qquad\qquad\qquad\quad\quad\\
+~\beta^h_{i,r}\frac{\partial\Psi}{\partial\lambda_{i,r}}~\text{U}\left(-\beta^h_{i,r}\frac{\partial\Psi}{\partial\lambda_{i,r}}\right)<0,\label{Th_8_2}
\ee
which implies that ${\bf v}_{\lambda}^h$ is a strictly descent direction, unless ${\bf v}_{\lambda}^h=0$.
\end{proof}

\begin{proof}[\bf Proof of Lemma~\ref{lem_eqv}]
By contradiction, assume that $\nabla_{\bf x}\Psi({\bf x}^h,{\bf \lambda}^h)\neq0$ and ${\bf v}^h=[{\bf v}_{x}^h;{\bf v}_{\lambda}^h]=0$. Consider some server $i$ for which $\nabla_{{\bf x}_i}\Psi({\bf x}^h,{\bf \lambda}^h)\neq 0$. The projection of $-\nabla_{{\bf x}_i}\Psi$ onto $\Omega_i$ can be found by solving the optimization in \eqref{Th8_0}. The problem in \eqref{Th8_0} is to minimize a convex function over an affine feasible region (see the definition of $\Omega_i$ in \eqref{projection_set}). The KKT conditions for this problem imply that:
\be
&&0 \le \varsigma_{i,r} \perp -({\bf v}_{x}^h)_i^T{\bf d}_{r}\ge0,~r\in\mathcal{R}_i^h:\lambda_{i,r}=0,\label{Th9_01}\quad\qquad\\
&&({\bf v}_{x}^h)_i^T{\bf d}_{r}=0,~r\in\mathcal{R}_i^h:\lambda_{i,r}>0,\label{Th9_02}\\
&&({\bf v}_{x}^h)_i + \nabla_{{\bf x}_i}\Psi + \sum_{r\in\mathcal{R}_i^h}\varsigma_{i,r}{\bf d}_{r}=0.\label{Th9_03}
\ee
Here, without loss of generality, we assume that the vectors in $\{{\bf d}_{r}\mid r\in\mathcal{R}_i^h \}$ are linear independent [Otherwise we may restrict the conditions in $\Omega_i$ to a subset $\mathcal{A}\subseteq\mathcal{R}_i^h$ for which $\{{\bf d}_{r}\mid r\in\mathcal{A}\}$ are linear independent.]. The projection of $-\nabla_{{\bf x}_i}\Psi$ onto $\Omega_i$ is zero, $({\bf v}_x^h)_i=0$, when:
\be
-\nabla_{{\bf x}_i}\Psi = \sum_{r\in\mathcal{R}_i^h}\varsigma_{i,r}{\bf d}_{r},\label{Th9_1}
\ee
for some scalars $\varsigma_{i,r}\in\mathbb{R}$ satisfying \eqref{Th9_01}.
In other words, ${\bf v}_{x}^h=0$ if and only if there exists some scalars $\varsigma_{i,r}\in\mathbb{R}$,
so that \eqref{Th9_1} is established for all servers with $\nabla_{{\bf x}_i}\Psi\neq 0$.

On the other hand, ${\bf v}^h_{\lambda}=0$ requires that $(\tilde{\bf v}_{\lambda}^h)_{i,r}=0$, for every $r\in\mathcal{R}_i^h$ with $\lambda_{i,r}>0$ (c.f. \eqref{lambda_feas_dir}).
However, $(\tilde{\bf v}_{\lambda}^h)_{i,r}=0$ only when $\partial\Psi/\partial\lambda_{i,r}=0$ and
$\beta_{i,r}^h=-{\bf d}_{r}^T\nabla_{{\bf x}_i}\Psi=0$ (c.f.~\eqref{lambda_non_feas_dir}).
it follows from \eqref{Th9_1} that:
\be
\beta_{i,r}^h=-{\bf d}_{r}^T\nabla_{{\bf x}_i}\Psi={\bf d}_{r}^T\sum_{r'\in\mathcal{R}_i^h}\varsigma_{i,r'}{\bf d}_{r'}. \label{Th9_2}
\ee
Since the vectors $\{{\bf d}_{r}\mid r\in\mathcal{R}_i^h \}$ are linear independent,
It follows from \eqref{Th9_2} that $\beta_{i,r}^h=0$ only when $\varsigma_{i,r}=0$.
Hence, $\varsigma_{i,r}=0$ for every $r\in\mathcal{R}_i^h$ with $\lambda_{i,r}>0$.
So, the right hand side of \eqref{Th9_1} would be zero under Assumption~\ref{assmp_degen}.
That is, $\nabla_{{\bf x}_i}\Psi=0,~\forall i$. This, however, violates the contradictory assumption.

Now let Assumption~\ref{assmp_degen} do not hold, so that for server $i$ with $\nabla_{{\bf x}_i}\Psi\neq 0$ there exists some resource $r\in\mathcal{R}_i^h$ with $\lambda_{i,r}=0$ and $\varsigma_{i,r}>0$ (c.f. \eqref{Th9_01}). Let $\bar{\mathcal{R}}_i^h:=\{r\in\mathcal{R}_i^h\mid \lambda_{i,r}=0\text{ and }\varsigma_{i,r}>0\}$.
Then, \eqref{Th9_1} implies that $\partial\Psi/\partial x_{n,i}<0$, for every user $n$ with $d_{n,r}>0,~r\in\bar{\mathcal{R}}_i^h$. On the other hand, $\partial\Psi/\partial x_{n,i}=0$ for every user $n\notin\mathcal{N}_i$ (c.f.~\eqref{partial_1}). Hence, \eqref{Th9_1} may not be established if there exists some user $n\notin\mathcal{N}_i$ with $d_{n,r}>0,~r\in\bar{\mathcal{R}}_i^h$, which is the case under Assumption~\ref{assmp_simple}. In other words, \eqref{Th9_1} may not be established under Assumption~\ref{assmp_simple}.
\end{proof}
.

\end{document}